\UseRawInputEncoding
\documentclass[journal]{IEEEtran}
\pdfoutput=1
\usepackage{algorithmic}
\usepackage{algorithm}
\usepackage{array}

\usepackage{subfigure}
\usepackage{amsmath,amsfonts}
\usepackage{amsthm}

\usepackage{textcomp}
\usepackage{stfloats}
\usepackage{url}
\usepackage{verbatim}
\usepackage{graphicx}
\usepackage{cite}
\usepackage{bm}
\usepackage{soul}
\usepackage{arydshln}
\usepackage{booktabs}
\usepackage{pifont}
\usepackage{amssymb}
\usepackage{tabularx}
\usepackage{microtype}
\usepackage{stfloats}
\usepackage{epstopdf}
\usepackage{makecell}
\usepackage{diagbox}

\allowdisplaybreaks[4]
\usepackage[implicit=false]{hyperref}
\usepackage{tikz,xcolor}

\newtheorem{remark}{Remark}
\newtheorem{proposition}{Proposition}

\begin{document}
\definecolor{lime}{HTML}{A6CE39}
\DeclareRobustCommand{\orcidicon}{%
    \begin{tikzpicture}
    \draw[lime, fill=lime] (0,0)
    circle [radius=0.16]
    node[white] {{\fontfamily{qag}\selectfont \tiny ID}};    \draw[white, fill=white] (-0.0625,0.095)
    circle [radius=0.007];    \end{tikzpicture}
    \hspace{-2mm}}
\foreach \x in {A, ..., Z}{%
    \expandafter\xdef\csname orcid\x\endcsname{\noexpand\href{https://orcid.org/\csname orcidauthor\x\endcsname}{\noexpand\orcidicon}}
    }
\hypersetup{hidelinks}
\newcommand{\orcidauthorA}{0000-0001-6215-6703}

\title{Anchor-points Assisted Uplink Sensing\\ in Perceptive Mobile Networks}
\author{Yanmo Hu,~\IEEEmembership{Student~Member,~IEEE}, J. Andrew Zhang,~\IEEEmembership{Senior~Member,~IEEE},\\Weibo Deng, Y. Jay Guo,~\IEEEmembership{Fellow,~IEEE}\vspace{-0.25em}
\thanks{Yanmo Hu and Weibo Deng are with the Key Laboratory of Marine Environmental Monitoring and Information Processing, Ministry of Industry and Information Technology, and also with the School of Electronic and Information Engineering, Harbin Institute of Technology. Email: 19B905008@stu.hit.edu.cn; dengweibo@hit.edu.cn. \emph{(Corresponding author: Weibo Deng.)}

J. A. Zhang and Y. J. Guo are with Global Big Data Technologies Centre, the University of Technology Sydney. Email: \{Andrew.Zhang; Jay.Guo\}@uts.edu.au.}}

\maketitle

\begin{abstract}
Uplink sensing in integrated sensing and communications (ISAC)
systems, such as Perceptive Mobile Networks, is challenging due to the clock asynchronism between transmitter and receiver.
Existing solutions typically require the presence of a dominating line-of-sight path and the knowledge of transmitter location at the receiver.
In this paper, relaxing these requirements, we propose a novel and effective uplink sensing scheme with the assistance of static anchor points. Two major algorithms are proposed in the scheme.
The first algorithm
estimates the relative timing and carrier frequency offsets due to clock asynchronism, with respect to those at a randomly selected reference snapshot.
Theoretical performance analysis is provided for the algorithm.
The estimates from the first algorithm are then used to compensate for the offsets and generate the angle-Doppler maps.
Using the maps, the second algorithm identifies the anchor points, and then locates the UE and dynamic targets.
Feasibility of UE localization is also analyzed.
Simulation results are provided and demonstrate the effectiveness of the proposed algorithms.
\end{abstract}

\begin{IEEEkeywords}
Integrated sensing and communication (ISAC), Perceptive mobile networks, uplink sensing, clock asynchronism, static anchor points, parameter estimation.
\end{IEEEkeywords}

\section{Introduction}\label{Section_I}
\subsection{Motivations and Backgrounds}\label{Section_I-A}
\IEEEPARstart{I}{SAC} \cite{9585321, 1011453411834, 9737357, 1011453351279} enables to share a majority of hardware and network infrastructure to integrate communication and sensing into one system, which currently serves as a candidate technology of 6G mobile networks.
With the application of ISAC in mobile networks,
the perceptive mobile networks (PMNs) \cite{9585321} can be achieved using either uplink \cite{9606921, 5393298} or downlink signals \cite{8376990}.
The user equipment (UE) and base station (BS) in uplink sensing can be regarded as the transmitter and receiver in a \emph{bistatic sensing} \cite{9585321}, respectively.
The major challenge in such a bistatic uplink sensing is that the
UE and BS typically use their local oscillators, leading to clock asynchronism.
Clock asynchronism may generate time-varying timing offset (TMO) and carrier frequency offset (CFO), causing ambiguity in delay and Doppler estimation and preventing from coherent processing of multiple discontinuous signals\cite{9585321, 9848428, 8443427, 9617146}.

Several techniques have been introduced to resolve the impact of clock asynchronism on uplink sensing \cite{9848428, 2018Widar, 9349171}.
Among them, the cross-antenna cross-correlation (CACC) technique is initially introduced for passive WiFi sensing \cite{2018Widar}.
Exploiting the fact that all antennas in an
array share the same clock oscillator,
CACC calculates the conjugate multiplication between signals from different receiving antennas,
successfully removing TMO and CFO.
However, the presence of multiple cross-product terms in CACC doubles the parameters to be estimated.
To address this problem,
a mirrored-MUSIC algorithm is proposed in \cite{9349171} to handle the basis vectors with mirror symmetry character.
It was shown that the mirrored-MUSIC algorithm can accurately detect multiple targets under the ISAC framework with relatively low computation complexity.
Another class of techniques,
similarly exploiting
the characteristic that the clock is identical across all the antennas in the BS, is the channel state information (CSI) ratio method \cite{Added_123777, 9645160}.
However, two fundamental assumptions used in these methods restrict their applications:
the position of UE is known,
and there exists a dominating line-of-sight (LOS) path between UE and BS.

A limited number of techniques have been reported to relax these requirements.
Without requiring the knowledge of the UE's position, recently,
\cite{10262009} locates dynamic targets by assuming that the TMO and CFO are noise-like and follow zero-mean Gaussian distribution.
By employing the Kalman filter (KF), \cite{10262009}
proposes a KF-based CSI enhancer to suppress the noise-like TMO.
Since the LOS from UE to BS is the shortest,
\cite{10262009} treats the smallest range value
as the LOS path range and then locates the UE.
However,
the performance of the scheme depends on the actual distribution of the TMO and CFO.
Regarding the requirement for LOS path,
\cite{8515231} analytically shows that
non-LOS (NLOS) paths alone can also provide sufficient
information for locating UE if and only if the estimates for angle-of-arrival (AOA), angle-of-departure (AOD), and range are sufficiently accurate.
Specifically, three or more NLOS paths are `possible'
for generating
unambiguous
estimation for
UE's
position and orientation.
Nevertheless,
\cite{8839839} proves that the localizability of using NLOS paths is
low under under the Boolean
model in 5G mm-wave systems.

Another potential method is to use anchor points, strong reflectors with known positions to the sensing receiver.
For wireless localization,
which involves the localization of signal-emitting transmitters,
it has been demonstrated that anchor points can efficiently reduce the system cost and provide a solution for target localization \cite{1433260, 8839839, 9301354},
particularly in scenarios where a line-of-sight (LOS) path is absent.
They can be classified into two groups: mobile and static anchor points.
The mobile anchor points \cite{1433260} are equipped with GPS, moving among the sensing area and periodically broadcasting  their current positions to help the localization.
The static anchor points \cite{9301354} generally have known positions and are usually strong scatterers, thus are easy to be distinguished in the echo signals.
Although anchor points have been demonstrated to be effective in assisting wireless localization,
work has yet to be reported for sensing in ISAC systems.
For uplink sensing, significant challenges associated with clock asynchronism need to be overcome before signals and parameters for anchor points can be exploited.

\subsection{Contributions and Organizations}\label{Section_I-B}

In this paper, we develop an anchor-points assisted uplink sensing scheme for ISAC systems.
The scheme will be presented by referring to a PMN involving a BS and UE, and it can also be applied to other networks.
We assume the existence of several static anchor points near the BS.
Our scheme comprises two core algorithms.
Algorithm \ref{Algorithm1} is designed to estimate relative TMO (RTMO) and relative CFO (RCFO) with respect to the values at a randomly selected snapshot.
It includes two sub-algorithms, coarse estimation and refined estimation.
Using the estimated RTMO and RCFO, Algorithm \ref{Algorithm1} compensates the TMO and CFO across multiple snapshots.
Algorithm \ref{Algorithm2} obtains sensing parameters and achieves localization for UE and dynamic targets with the assistance of anchor points.
Notably, our scheme does not require to know the UE's location or the existence of a line-of-sight (LOS) path between the BS and UE,
thus addressing the limitations of existing technologies.
It is very promising for practical uplink sensing in PMN, and general ISAC systems.

Our main contributions are summarized below:
\begin{itemize}
  \item
  We propose a coarse estimation algorithm for RTMO and RCFO in Algorithm \ref{Algorithm1},
  which works efficiently when the power of static paths is larger than that of dynamic ones.
  We transform the clock asynchronism problem into a single-tone parameter estimation problem solved by maximum likelihood estimation (MLE).
  This algorithm works with both continuous or discontinuous CSI measurements either in the time or frequency domain.

  \item
  We propose an iterative algorithm in Algorithm \ref{Algorithm1} to obtain refined estimates of RTMO and RCFO, using the coarse results as the initial value.
  The main purpose is to recover the Doppler of static objects, i.e., $f_D = 0$.
  We also derive the bias and theoretical RMSE of the estimation errors and prove that they have negligible effects on the subsequent parameter estimation for dynamic targets and anchor points.
  The theoretical results also demonstrate that Algorithm \ref{Algorithm1} is robust under any values of RTMO and RCFO.

  \item
  We propose Algorithm \ref{Algorithm2} to obtain sensing parameter estimations, pinpoint anchor points, and locate both UE and dynamic targets.
  We analytically show that this algorithm works when the number of anchor points is not less than 2.
  We also analyze the feasibility of localization in relation to the locations of anchor points and UE.

\end{itemize}

The rest of this paper is organized as follows.
Section \ref{II-A} introduces the uplink sensing model without the LOS path.
The RTMO and RCFO estimation algorithm is presented in Section \ref{III},
including the coarse and refined estimation algorithms in Section \ref{III-A} and \ref{III-B}, respectively.
The performance analysis is given in Section \ref{Performance}.
Section \ref{IV} presents the localization algorithm for UE and dynamic targets.
Simulation results are shown in Section \ref{V}, and conclusions are drawn in Section \ref{VI_conclusion}.

$Notation$:
$\odot $ denotes the Hadamard product;
$\circledast $ denotes the convolution;
${{\mathbf{A}}^{T}}$, ${{\mathbf{A}}^{*}}$ and ${{\mathbf{A}}^{H}}$ represent the transpose, conjugate and conjugate transpose of $\mathbf{A}$, respectively;
${{\left[ \mathbf{a}  \right]}_{n}}$ is the $n$-th element of vector $\mathbf{a}$, and ${{\left[ \mathbf{A}  \right]}_{m,n}}$ is the $\left(m,n \right)$-th element of matrix $\mathbf{A}$;
$\angle \left\{ a \right\}$ denotes the argument of the complex number;
${{\mathcal{F}}_{2}}$ is the two-dimensional discrete fourier transform;
${{\mathbf{I}}_{M}}\in {{\mathbb{R}}^{M\times M}}$ is identity matrix;
${{\mathbf{1}}_{M\times 1}}\in {{\mathbb{R}}^{M\times 1}}$ is the vector with all the elements are ones,
while ${{\mathbf{0}}_{M\times 1}}\in {{\mathbb{R}}^{M\times 1}}$ is the vector filled by zeros.
In particular, ${{\mathbf{1}}_{K}}\in {{\mathbb{R}}^{K\times K}}$ is the square matrix;
${{\mathbf{A}}^{\dagger }}={{\left( {{\mathbf{A}}^{T}}\mathbf{A} \right)}^{-1}}{{\mathbf{A}}^{T}}$ is the pseudo-inverse of $\mathbf{A}$;
${\mathbf{e}}_{M, m}\in {{\mathbb{R}}^{M\times 1}}$ is a vector with the $m$-th element equal to 1 and others equal to 0, where $m=0,\cdots,M-1$.

\section{Signal Model of Uplink Sensing}\label{II-A}

\begin{figure}[!t]
\centering
\includegraphics[width=3.2in]{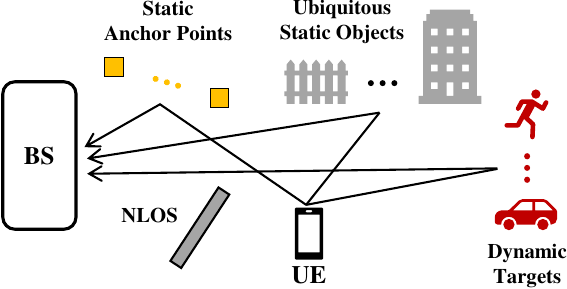}
\caption{The uplink sensing model without the LOS path.}
\label{Fig1}
\end{figure}

The uplink communication and sensing model is illustrated in Fig. \ref{Fig1}.
We consider an ISAC system with one static UE and one static BS.
The UE has $M_t$ omnidirectional transmitter antennas, while the BS is equipped with a uniform linear array with $M_r$ receiver antennas, and the antenna interval is half-wavelength.
The uplink signal is received by the BS after being reflected by two types of objects: dynamic and static.
We consider the typical case when most of reflected signals are from static objects, having zero Doppler.
The number of dynamic objects is relatively low, but they are the targets of interest for sensing, such as human targets, cars, etc.

Similar to existing works on bistatic sensing \cite{2018Widar, 9349171}, we consider the existence of clock offset between the UE transmitter and BS receiver, which causes sensing ambiguity and prevents from coherent processing of discontinuous channel measurements \cite{{9848428}}.
However, unlike in most of these works, we consider
\emph{a more challenging scenario, where there is no LOS path between the UE and BS, and UE's location is unknown to the BS}.
For a single BS where multiple antennas are closely located,
multi-station-based techniques such as TDOA cannot be applied to solve the clock asynchronism problem.
Hence, we emphasize the necessity of leveraging additional auxiliary information to resolve the clock offset in this scenario.
Failing to do so could result in ambiguous range and Doppler shift estimates for the dynamic targets.
The following two fundamental assumptions made in this paper can enable sensing in the absence of the LOS path between the UE and BS and without knowing the location of the UE:
\begin{itemize}
  \item
  Assumption One (A1): The total power of the electromagnetic wave reflected from the static objects is stronger than that of the dynamic targets;
  \item
  Assumption Two (A2): Assuming that the precise position of $L_a \ge 1$ static anchor points \emph{close to} the BS is known, and signals reflected from them are strong and can be separated from other multipath signals by the measured AoA and power.
\end{itemize}

For A1, in many applications,
there are often a large number of static objects/reflectors in the sensing environment, such as buildings and plants outdoors and walls and furnitures indoors.
They can lead to stronger reflected signals compared to dynamic targets, particularly when these dynamic targets are not very close to the transmitter or receiver.
A1 provides a zero Doppler reference against the clock asynchronism influence,
and is mainly used for estimating the RCFO and RTMO, as will be detailed in Section III.

For A2, the anchor points can be from either existing static objects or specifically installed strong reflectors, such as Corner Reflectors, positioned near the receiver.
The positions for anchor points can be obtained with ease, via installation record or long-period averaging.
Note that we only consider passive reflective anchor points in this paper,
although it is possible to extend our scheme to more advanced anchor points such as reconfigurable intelligent surface (RIS) with beam steering capabilities.
A2 is mainly used to remove a timing ambiguity to generate the estimates of the absolute locations of UE and dynamic targets, as will be discussed in Section \ref{IV}.

In what follows, we establish the signal model of uplink sensing based on the above assumptions.
Our model is an abstract of 5G mobile signals.
The demodulation reference signals (DMRS) for the
physical uplink shared channel in 5G NR \cite{5GNR211, 8928165, 9779639} have fixed and known signal values to the BS,
which can be effectively used for uplink sensing.
Suppose that
the subcarrier spacing of the OFDM signal is $\Delta f $
and the number of subcarriers is $N$.
Thus, the bandwidth is $B = N \Delta f$.
Assume that the DMRS occupies the same subcarriers across OFDM symbols
 and it is periodic
 in the time domain,
and the interval between two adjacent OFDM symbols is $T_0$.
We use $K$ to represent the number of used DMRS symbols
for sensing.
Since the clocks for BS and UE are not locked,
transmissions can cause time-varying phase shift, TMO and CFO to received DMRS symbols \cite{9585321}.
We use $e^{j\beta_k}$, $ {{\tau }_{o,k}}$ and ${{f}_{o,k}} $ ($k = 0, \cdots, K - 1$) to represent the phase shift, TMO and CFO, respectively.
Note that here the CFO is the small residual CFO after CFO estimation and compensation, which are typically applied in the time domain before signal processing in the frequency domain.
Thus, the phase shift caused by the CFO can typically be approximated as a constant within one OFDM symbol.

Assume that $L$ targets distribute in the space, with the propagation delay $\tau_{\ell}$, Doppler $f_{D, \ell}$, AOD $\varphi_{\ell}$ and AOA $\theta_{\ell}$, where $\ell=0,\cdots,L-1$ is target index.
We assume that the number of dynamic targets and static objects is $L_d$ and $L_s$, respectively, and we have $L = L_d + L_s$.
For uplink sensing, the CSI for the $n$-th subcarrier ($n = 0, \cdots, N - 1$) of the $k$-th OFDM symbol between the $m_t$-th antenna of UE and the $m_r$-th antenna of BS is given by \cite{9585321, 9848428, 10262009}:
\begin{align}\label{formula_original}
\nonumber y_{n,k}^{\left( {{m}_{t}},{{m}_{r}} \right)} = {}&{{e}^{j{{\beta }_{k}}}}\sum\limits_{\ell =0}^{L-1}{{{\xi }_{\ell }}{{e}^{-j2\pi n\left( {{\tau }_{\ell }}+{{\tau }_{o,k}} \right)\Delta f}}}{{e}^{j2\pi k\left( {{f}_{D,\ell }}+{{f}_{o,k}} \right){{T}_{0}}}}  \\
&\times {{e}^{j\frac{2\pi }{\lambda }{{m}_{r}}d\sin {{\theta }_{\ell }}}}{{e}^{j\frac{2\pi }{\lambda }{{m}_{t}}d\sin {{\varphi }_{\ell }}}}+z_{n,k}^{\left( {{m}_{t}},{{m}_{r}} \right)},
\end{align}
where $\lambda$ denotes the wavelength, ${{\xi }_{\ell }}$ denotes the value of the complex-valued reflection coefficient and the propagation attenuation, and
${{z}_{n,k}^{\left(m\right)}} \in {\mathbb{C}}$ is the additive white Gaussian noise matrix with ${{z}_{n,k}^{\left(m\right)}}\sim \mathcal{C}\mathcal{N}\left( 0,\sigma _{n}^{2} \right)$.
The power attenuation of the NLOS path is ${\frac{{{\lambda }^{2}}{{\sigma }_{x}}}{{{\left( 4\pi  \right)}^{3}}R_{T}^{2}R_{R}^{2}}}$, where $\sigma_x$, $R_T$ and $R_R$ denote the reflecting factor, the transmitter-to-target range and the receiver-to-target range, respectively.
Suppose that the noise is statistically independent at different subcarriers, OFDM symbols and receiver antennas.

In (\ref{formula_original}), ${{e}^{j{{\beta}_{k}}}}$ is the random phase in the $k$-th OFDM symbol.
This term has little influence on the communication, as it can be absorbed by channel estimation.
While in sensing applications,  ${{e}^{j{{\beta}_{k}}}}$ in different OFDM symbols must be equal or known,
otherwise the ability of using complex signals to estimate the Doppler will be lost.
However, ${{e}^{j{{\beta}_{k}}}}$ changes randomly across different OFDM symbols and is unknown both in BS and UE.
Note that the ${{e}^{j{{\beta}_{k}}}}$ in (\ref{formula_original}) is only dependent of the OFDM symbol index $k$, similar to the CFO term.
In this paper, we combine the random phase and CFO and eliminate their impact simultaneously.

Now, we rewrite (\ref{formula_original}) as:
\begin{align}\label{1}
 y_{n,k}^{\left( m \right)}= {{e}^{j{{{\tilde{f}}}_{o,k}}}}\sum\limits_{\ell =0}^{L-1}{{{\xi }_{\ell }}{{e}^{-j2\pi n\left( {{\tau }_{\ell }}+{{\tau }_{o,k}} \right)\Delta f}}{{e}^{j2\pi k{{f}_{D,\ell }}{{T}_{0}}}}{a_{\ell}}{\left( m \right)}}+{{z}_{n,k}^{\left(m\right)}},
\end{align}
where we simplify the angle term for ease of derivation: ${{a}_{\ell}}\left( m \right)={{e}^{j\frac{2\pi }{\lambda }{{m}_{r}}d\sin {{\theta }_{\ell }}}}{{e}^{j\frac{2\pi }{\lambda }{{m}_{t}}d\sin {{\varphi }_{\ell }}}}$ and $m=0,\cdots,{M_t}{M_r} - 1$, and ${{{\tilde{f}}}_{o,k}}\triangleq {{\beta }_{k}}+2\pi k{{f}_{o,k}}{{T}_{0}}$
combines the random phase and the CFO.

\section{Estimation of Relative CFO and Relative TMO}\label{III}
In this section, the derivation is made only under the assumption of A1.
We select reference values at a give time as the benchmarks to estimate the relative values of other CFO and TMO in this section, and then propose methods to estimate the reference values in Section \ref{IV}.
Let ${{\tau }_{o,0 }}$ and ${\tilde{f }_{o,0 }}$ denote the respective reference values of TMO and CFO.
The RTMO and RCFO values are defined by ${{\tau }_{o,k}^R} = {{\tau }_{o,k}} - {{\tau }_{o,0}}$ and ${\tilde{f }_{o,k}^R} = {\tilde{f }_{o,k}} - {\tilde{f }_{o,0}}$, $k=1,\cdots, K-1$, respectively.
Note that ${{\tau }_{o,0}^R} = {0}$ and ${\tilde{f }_{o,0}^R} = 0$.

Based on A1, our main idea is to make full use of the characteristic of static objects and solve an optimization function to obtain the RTMO and RCFO.
Note that the received signal in each received antenna is processed independently in this section.

\subsection{Estimation of RTMO and RCFO: Coarse Estimation}\label{III-A}
Now, we collect the data from each subcarrier, i.e., ${{\mathbf{y}}^{\left(m\right)}_{k}}={{\left[ {{y}^{\left(m\right)}_{0,k}},\cdots ,{{y}^{\left(m\right)}_{N-1,k}} \right]}^{T}}\in {{\mathbb{C}}^{N\times 1}}$ with $k = 0, \cdots, K-1$, and construct the following vector based on (\ref{1}):
\begin{align}\label{2}
{{\mathbf{y}}^{\left(m\right)}_{k}}=
{{e}^{j{{{\tilde{f}}}_{o,k}}}}\text{diag}\left\{ {{\mathbf{t}}_{k}} \right\}\mathbf{\Gamma }\left[\bm{\xi }\odot {{\mathbf{d}}_{k}} \odot \mathbf{a}\left( m\right) \right],
\end{align}
where
$\mathbf{\Gamma }\in {{\mathbb{C}}^{N\times L}}$, ${{\left[ \mathbf{\Gamma } \right]}_{n,\ell }}={{e}^{-j2\pi n{{\tau }_{\ell }}\Delta f}}$ denotes the matrix containing all targets' delay;
${{\mathbf{t}}_{k}}={{\left[ 1,{{e}^{-j2\pi {{\tau }_{o,k}}\Delta f}},\cdots ,{{e}^{-j2\pi \left( N-1 \right){{\tau }_{o,k}}\Delta f}} \right]}^{T}}\in {{\mathbb{C}}^{N\times 1}}$ represents the TMO vector;
${{\bm{\xi }}}={{\left[ {{\xi }_{0}},\cdots ,{{\xi }_{L-1}} \right]}^{T}}\in {{\mathbb{C}}^{L\times 1}}$ denotes the vector of complex reflection coefficient;
${{\mathbf{d}}_{k}}={{\left[ {{e}^{j2\pi k{{f}_{D,0}}{{T}_{0}}}},\cdots ,{{e}^{j2\pi k{{f}_{D,L-1}}{{T}_{0}}}} \right]}^{T}}\in {{\mathbb{C}}^{L\times 1}}$ denotes the Doppler phase in the $k$-th OFDM block;
and
$\mathbf{a}\left( m \right)={{\left[ {{a}_{0}}\left( m \right),\cdots ,{{a}_{L-1}}\left( m \right) \right]}^{T}}\in {{\mathbb{C}}^{L\times 1}}$.

Different from the previous methods
which employ the conjugate product operation between different antennas, such as the CACC \cite{2018Widar} and the CSI-ratio
method \cite{Added_123777}, the coarse estimation calculates the conjugate product between different OFDM symbols.
Construct the following formula:
\begin{align}\label{3}
 \nonumber\bm{\varpi} _{k,0}^{\left( m \right)} = {}& \text{diag}\left\{ \left(\mathbf{y}^{\left(m\right)}_{0}\right)^{*} \right\}{{\mathbf{y}}^{\left(m\right)}_{k}} \\
 \nonumber  = {}& {{e}^{j\left( {{{\tilde{f}}}_{o,k}}-{{{\tilde{f}}}_{o,0}} \right)}}\text{diag}\left\{ {{\mathbf{t}}_{k}}\odot \mathbf{t}_{0}^{*} \right\}\\
 \nonumber &\times \text{diag}\left\{ {{\mathbf{\Gamma }}^{*}}{{\left[\bm{\xi }\odot {{\mathbf{d}}_{0}} \odot \mathbf{a}\left( m \right) \right]}^{*}} \right\}\mathbf{\Gamma }\left[\bm{\xi }\odot {{\mathbf{d}}_{k}} \odot \mathbf{a}\left( m \right) \right] \\
 \nonumber   \approx {}& {{e}^{j {\tilde{f }_{o,k}^R} }}\text{diag}\left\{ {{\mathbf{t}}_{k}}\odot \mathbf{t}_{0}^{*} \right\}\\
\nonumber & \times \text{diag}\left\{ {{\mathbf{\Gamma }}^{*}}{{\left[ \bm{\xi }\odot\mathbf{a}\left( m \right) \right]}^{*}} \right\}\mathbf{\Gamma }\left[ \bm{\xi }\odot\mathbf{a}\left( m \right) \right] \\
  = {}& {{e}^{j {\tilde{f }_{o,k}^R} }}\text{diag}\left\{ \mathbf{t}_{k}^{R} \right\}{{\mathbf{h}}^{\left(m\right)}},
\end{align}
where
$\bm{\varpi} _{k,0}^{\left( m \right)} \in {{\mathbb{C}}^{N\times 1}}$, $k = 1, \cdots, K - 1$ is the result of phase difference between the $0$-th OFDM symbol and the $k$-th symbol,
$\mathbf{t}_{k}^{R}\triangleq {{\left[ 1,{{e}^{-j2\pi \tau _{o,k}^{R}\Delta f}},\cdots ,{{e}^{-j2\pi \left( N-1 \right)\tau _{o,k}^{R}\Delta f}} \right]}^{T}}\in {{\mathbb{C}}^{N\times 1}}$
contains the RTMO value,
${{\mathbf{h}}^{\left(m\right)}}\triangleq\text{diag}\left\{ {{\mathbf{\Gamma }}^{*}}{{\left[ \bm{\xi }\odot \mathbf{a}\left( m \right) \right]}^{*}} \right\}\mathbf{\Gamma }\left[ \bm{\xi } \odot \mathbf{a}\left( m \right) \right]\in {{\mathbb{R}}^{N\times 1}}$ and each element in ${{\mathbf{h}}^{\left(m\right)}}$ is a positive and real number.
Here, the approximation in the third step is based on A1.

\allowdisplaybreaks[4]

The last line of (\ref{3}) demonstrates that the initial phase and the frequency of $\bm{\varpi} _{k,0}^{\left( m \right)}$ are the $k$-th RCFO and RTMO, respectively.
Herein, we normalize $\bm{\varpi} _{k,0}^{\left( m \right)}$ in order to eliminate the impact from
the amplitude
${{\mathbf{h}}^{\left(m\right)}}$:
\begin{align}\label{4}
 {{\left[ \bm{\tilde{\varpi} }_{k,0}^{\left( m \right)} \right]}_{n}}={{{{\left[ \bm{\varpi} _{k,0}^{\left( m \right)} \right]}_{n}}}/{\left| {{\left[ \bm{\varpi} _{k,0}^{\left( m \right)} \right]}_{n}} \right|}}.
\end{align}
where $n=0,\cdots,N - 1$.
It is worth noting that (\ref{4}) is a typical single-tone parameter estimation problem \cite{1055282}, and the MLE can be employed
to estimate the unknown parameters.
However, the noise component in (\ref{3}) does not follow the Gaussian distribution.
Here, we still hypothesize that the noise
obeys Gaussian distribution
in order to facilitate the application of MLE.
Although the estimation is not sufficiently accurate,
it is acceptable since we only find the initial estimates.
The approximated MLE formulas are given by:
\begin{subequations}\label{MLE_fomular_RTMO_RCFO}
\begin{align}
  \hat{\tau }_{o,k}^{R{\left( m \right)}}&=\underset{\tau _{o,k}^{R}}{\mathop{\arg \max }}\,\left| \sum\limits_{n=0}^{N-1}{{{e}^{j2\pi n\tau _{o,k}^{R}\Delta f}}{{\left[ \bm{\tilde{\varpi }}_{k,0}^{\left( m \right)} \right]}_{n}}} \right|, \label{MLE_fomular_RTMO} \\
\hat{\tilde{f}}_{o,k}^{R\left( m \right)}&=\angle \left\{ \sum\limits_{n=0}^{N-1}{{{e}^{j2\pi n\hat{\tau }_{o,k}^{R\left( m \right)}\Delta f}}{{\left[ \bm{\tilde{\varpi }}_{k,0}^{\left( m \right)} \right]}_{n}}} \right\}. \label{MLE_fomular_RCFO}
\end{align}
\end{subequations}
Note that the fast Fourier transform (FFT) algorithm is conducive to reducing the computation complexity of (\ref{MLE_fomular_RTMO}).

After performing (\ref{MLE_fomular_RTMO_RCFO}) on the received signal at each antenna,
we finally obtain ${M_t}{M_r}$ estimates of RTMO and ${M_t}{M_r}$ estimates of RCFO.
Note that to reduce complexity, we can compute and utilize $M$ out of ${M_t}{M_r}$ estimates, where $M$ can be determined by the computational complexity.
Thus, the final coarse estimates of RTMO and RCFO are given by:
\begin{align}\label{final_coarse_RTMO_RCFO}
 \hat{\tau }_{o,k}^{R}={\frac{1}{{M}}}\sum\limits_{m=0}^{M-1}{\hat{\tau }_{o,k}^{R\left( m \right)}}, \
  \hat{\tilde{f}}_{o,k}^{R}={\frac{1}{M}}\sum\limits_{m=0}^{M-1}{\hat{\tilde{f}}_{o,k}^{R\left( m \right)}},
\end{align}
where $M \le {M_t}{M_r}$.
We set $\hat{\tau }_{o,0}^{R} = 0$ and $\hat{\tilde{f}}_{o,0}^{R} = 0$.
It is worth noting that $\hat{\tau }_{o,k}^{R}$ and $\hat{\tilde{f}}_{o,k}^{R}$ are respectively equal to ${\tau }_{o,k}^{R}$ and ${\tilde{f}}_{o,k}^{R}$ when neither the dynamic targets nor the noise exists.

By utilizing (\ref{final_coarse_RTMO_RCFO}), we can compensate the data in (\ref{1}):
\begin{align}\label{first_compensate}
 \nonumber \tilde{y}_{n,k}^{\left( m \right)}= {}&{{e}^{-j\hat{\tilde{f}}_{o,k}^{R}}}{{e}^{j2\pi n\hat{\tau }_{o,k}^{R}\Delta f}}y_{n,k}^{\left( m \right)} \\
 \nonumber  = {}&{{e}^{j\left( {{{\tilde{f}}}_{o,k}}-\hat{\tilde{f}}_{o,k}^{R}+{{{\tilde{f}}}_{o,0}}-{{{\tilde{f}}}_{o,0}} \right)}}{{e}^{-j2\pi n\left( {{\tau }_{o,k}}-\hat{\tau }_{o,k}^{R}+{{\tau }_{o,0}}-{{\tau }_{o,0}} \right)\Delta f}}\\
 \nonumber &\times \sum\limits_{\ell =0}^{L-1}{{{\xi }_{\ell }}{{e}^{-j2\pi n{{\tau }_{\ell }}\Delta f}}{{e}^{j2\pi k{{f}_{D,\ell }}{{T}_{0}}}}{a_{\ell}}{\left( m \right)}} \\
\nonumber  ={}&{{e}^{j\Delta \tilde{f}_{o,k}^{R}}}{{e}^{-j2\pi n\Delta \tau _{o,k}^{R}\Delta f}}\\
  &\times \sum\limits_{\ell =0}^{L-1}{{{{\tilde{\xi }}}_{\ell }}{{e}^{-j2\pi n\left( {{\tau }_{\ell }}+{{\tau }_{o,0}} \right)\Delta f}}{{e}^{j2\pi k{{f}_{D,\ell }}{{T}_{0}}}}{a_{\ell}}{\left( m \right)}},
\end{align}
where $k=0,\cdots,K-1$; $\Delta \tau _{o,k}^{R}=\left( {{\tau }_{o,k}}-{{\tau }_{o,0}} \right)-\hat{\tau }_{o,k}^{R}$
and
$\Delta \tilde{f}_{o,k}^{R}=\left( {{{\tilde{f}}}_{o,k}}-{{{\tilde{f}}}_{o,0}} \right)-\hat{\tilde{f}}_{o,k}^{R}$
are the errors of coarse RTMO and RCFO, respectively, and are both small;
${{\tilde{\xi }}_{\ell }}={{e}^{j{{{\tilde{f}}}_{o,0}}}}{{\xi }_{\ell }}$.

\begin{remark}\label{remark_coarse_RTMO_RCFO}
In the last line of (\ref{first_compensate}), ${{e}^{j{{{\tilde{f}}}_{o,0}}}}$ and ${{e}^{-j2\pi n{{\tau }_{o,0}}\Delta f}}$ are the residual errors after the compensation.
Since ${{{\tilde{f}}}_{o,0}}$ is a constant and is independent of $k$, $n$ and $m$,
${{e}^{j{{{\tilde{f}}}_{o,0}}}}$ has no influence on parameter (Doppler, delay and angle) estimation.
Thus, we can merge ${{e}^{j{{{\tilde{f}}}_{o,0}}}}$ with ${{\xi }_{\ell }}$ and define
${{\tilde{\xi }}_{\ell }}={{e}^{j{{{\tilde{f}}}_{o,0}}}}{{\xi }_{\ell }}$.
However, ${{e}^{-j2\pi n{{\tau }_{o,0}}\Delta f}}$ is dependent of $n$.
It will lead to the delay ambiguity \cite{9585321}, and we will present how to remove its impact in Section \ref{IV}.
In what follows, we discuss how to estimate $\Delta \tau _{o,k}^{R}$
and
$\Delta \tilde{f}_{o,k}^{R}$.
\end{remark}

\subsection{Estimation of RTMO and RCFO Errors: Refined Estimation}\label{III-B}
For ease of illustration, we reconstruct a set of new vectors along the time domain based on (\ref{first_compensate}),
i.e., ${{\tilde{\mathbf{y}}}^{\left(m\right)}_{n}}={{\left[ {\tilde{y}_{n,0}^{\left(m\right)}},\cdots ,{\tilde{y}_{n,K-1}^{\left(m\right)}} \right]}^{T}}\in {{\mathbb{C}}^{K\times 1}}$, $n=0,\cdots,N-1$:
\begin{align}\label{7}
\mathbf{\tilde{y}}_{n}^{\left( m \right)}=\text{diag}\left\{ \Delta {{\mathbf{f}}^{R}}\odot \Delta \bm{\gamma }_{n}^{R} \right\}\mathbf{D}\left[ \bm{\tilde{\xi }}\odot {{{\mathbf{\tilde{s}}}}_{n}}\odot \mathbf{a}\left( m \right) \right],
\end{align}
where
${{\mathbf{D}}}\in {{\mathbb{C}}^{K\times L}}$, ${{\left[ {{\mathbf{D}}} \right]}_{k,\ell }}={{e}^{j2\pi k{{f}_{D,\ell }}{{T}_{0}}}}$  represents the Doppler matrix composed of the Doppler phases of all targets,
${\Delta{\mathbf{\tilde{f}}}^R}={{\left[ 1,{{e}^{j{\Delta \tilde{f}_{o,1}^{R}}}},\cdots ,{{e}^{j{\Delta \tilde{f}_{o,K-1}^{R}}}} \right]}^{T}}\in {{\mathbb{C}}^{K\times 1}}$
and
$ \Delta \bm{\gamma }_{n}^{R}=\left[ 1,{{e}^{-j2\pi n\Delta \tau _{o,1}^{R}\Delta f}}, \right. $
$ {{\left. \cdots ,{{e}^{-j2\pi n\Delta \tau _{o,K-1}^{R}\Delta f}} \right]}^{T}}\in {{\mathbb{C}}^{K\times 1}} $
denote the vector of the estimation errors of coarse RCFO and RTMO, respectively,
${{\mathbf{s}}_{n}}={{\left[ {{e}^{-j2\pi n\left( {{\tau }_{0}}+{{\tau }_{o,0}} \right)\Delta f}},\cdots ,{{e}^{-j2\pi n\left( {{\tau }_{L-1}}+{{\tau }_{o,0}} \right)\Delta f}} \right]}^{T}}\in {{\mathbb{C}}^{L\times 1}}$.
The updated complex reflection coefficient vector is defined by $\bm{\tilde{\xi }}={{e}^{j{{{\tilde{f}}}_{o,0}}}}\bm{\xi }$.

As aforementioned, the targets are composed of dynamic targets and static objects, with numbers $L_d$ and $L_s$, respectively.
Equation (\ref{7}) can thus be divided into two parts:
\begin{align}\label{two_parts}
\nonumber   &\mathbf{\tilde{y}}_{n}^{\left( m \right)}=\text{diag}\left\{ \Delta {{\mathbf{f}}^{R}}\odot \Delta \bm{\gamma }_{n}^{R} \right\}\left[ {{\mathbf{D}}_{d}},{{\mathbf{D}}_{s}} \right] \left[ \begin{matrix}
   {{{\bm{\tilde{\xi }}}}_{d}}\odot  {{{\mathbf{\tilde{s}}}}_{d,n}} \odot \mathbf{a}_d\left( m \right)\\
   {{{\bm{\tilde{\xi }}}}_{s}} \odot {{{\mathbf{\tilde{s}}}}_{s,n}} \odot \mathbf{a}_s\left( m \right)\\
\end{matrix} \right]  \\
\nonumber  &=\text{diag}\left\{\Delta {{\mathbf{f}}^{R}}\odot \Delta \bm{\gamma }_{n}^{R}  \right\}{{\mathbf{D}}_{d}}\left[ {{{\bm{\tilde{\xi }}}}_{d}}\odot {{{\mathbf{\tilde{s}}}}_{d,n}}\odot \mathbf{a}_d\left( m \right) \right]\\
  &+\text{diag}\left\{ \Delta {{\mathbf{f}}^{R}}\odot \Delta \bm{\gamma }_{n}^{R} \right\}{{\mathbf{1}}_{K\times L_s}}\left[ {{{\bm{\tilde{\xi }}}}_{s}}\odot {{{\mathbf{\tilde{s}}}}_{s,n}}\odot \mathbf{a}_s\left( m \right) \right],
\end{align}
where the matrix or vector with subscripts `$d$' and `$s$' belong to the dynamic targets and static objects, respectively.
Note that the Doppler frequency of static objects is zero, i.e., ${{\mathbf{D}}_{s}}={{\mathbf{1}}_{K\times L_s}}$.

Although the phase of $\Delta {{\mathbf{f}}^{R}}\odot \Delta \bm{\gamma }_{n}^{R}$ is small, its impact cannot be overlooked.
In (\ref{two_parts}), the static Doppler matrix is contaminated by $\Delta {{\mathbf{f}}^{R}}$ and $\Delta \bm{\gamma }_{n}^{R}$, and thus, the static object energy will diffuse into the non-zero Doppler areas, and the energy of zero Doppler frequency will decrease.
Our main idea is to utilize the Doppler characteristic of static objects and guarantee that the zero Doppler frequency possesses the maximum energy based on A1.
The static object power can be expressed as:
\begin{align}\label{approximation_formula_11}
\nonumber & {{\mathcal{P}}_{s}}=\frac{1}{NK}\sum\limits_{n=0}^{N-1}{\sum\limits_{k=0}^{K-1}{{{\left| h_{s,n}^{\left( m \right)} \right|}^{2}}}} \\
\nonumber & =\frac{1}{N}\sum\limits_{n=0}^{N-1}{{{\left| \frac{1}{K}{{\left( \Delta {{\mathbf{f}}^{R}}\odot \Delta \bm{\gamma }_{n}^{R} \right)}^{H}}\text{diag}\left\{ \Delta {{\mathbf{f}}^{R}}\odot \Delta \bm{\gamma }_{n}^{R} \right\}{{\mathbf{1}}_{K\times 1}}h_{s,n}^{\left( m \right)} \right|}^{2}}} \\
 & \approx \frac{1}{N}\sum\limits_{n=0}^{N-1}{{{\left| \frac{1}{K}{{\left( \Delta {{\mathbf{f}}^{R}}\odot \Delta \bm{\gamma }_{n}^{R} \right)}^{H}}\mathbf{\tilde{y}}_{n}^{\left( m \right)} \right|}^{2}}},
\end{align}
where $h_{s,n}^{\left( m \right)} \triangleq {{\mathbf{1}}_{1\times L_s}}\left[ {{{\bm{\tilde{\xi }}}}_{s}}\odot {{{\mathbf{\tilde{s}}}}_{s,n}}\odot \mathbf{a}_s\left( m \right) \right] \in \mathbb{C}^{1}$ denotes the static value at the $n$-th subcarrier of the $m$-th antenna.
We now construct the cost function
$\mathcal{L}\left( \mathbf{c},\bm{\kappa } \right)=\sum\limits_{n=0}^{N-1}{{{\left| \left( {{e}^{j{{\mathbf{c}}^{T}}}}\odot {{e}^{-jn{{\bm{\kappa }}^{T}}}} \right)\mathbf{\tilde{y}}_{n}^{\left( m \right)} \right|}^{2}}}$,
where ${{{\mathbf{c}}},{{\bm{\kappa }}}\in {{\mathbb{R}}^{K\times 1}}}$,
and it can be easily verified that
$NK^2{{P}_{r}}\approx \underset{\mathbf{c},\bm{\kappa }}{\mathop{\max }}\,\mathcal{L}\left( \mathbf{c},\bm{\kappa } \right) $
when
$\exp \left( j\mathbf{c} \right)=\Delta {{\mathbf{f}}^{R}}$
and
$\exp \left( -jn\bm{\kappa } \right)=\Delta \bm{\gamma }_{n}^{R}$.

We have now constructed the relationship between $\mathcal{L}\left( \mathbf{c},\bm{\kappa } \right)$
and static path power ${{\mathcal{P}}_{s}}$
based on A1,
and $\Delta {{\mathbf{f}}^{R}}$ and $ \Delta \bm{\gamma }_{n}^{R}$ can be estimated by
maximizing $\mathcal{L}\left( \mathbf{c},\bm{\kappa } \right)$.
As aforementioned, the received
signal at each received antenna is processed independently.
We can formulate the following optimization problem to obtain $\Delta {{\mathbf{f}}^{R}}$ and $ \Delta \bm{\gamma }_{n}^{R}$, which is energy-based:
\begin{align}\label{9}
\underset{{{\mathbf{c}}^{\left( m \right)}},{{\bm{\kappa }}^{\left( m \right)}}\in {{\mathbb{R}}^{K\times 1}}}{\mathop{\max }} \,\mathcal{L}\left( {{\mathbf{c}}^{\left( m \right)}},{{\bm{\kappa }}^{\left( m \right)}} \right) \ \,
s.t. \ c_{0}^{\left( m \right)}=\kappa _{0}^{\left( m \right)}=0.
\end{align}

Note that the cost function in (\ref{9}) is nonlinear for both $\mathbf{{c}}^{\left( m \right)}$ and $\bm{\kappa}^{\left( m \right)}$,
so that the closed-form solution does not exist.
The Newton-Raphson method \cite{ArraySignalProcessing} can be utilized to maximize the cost function and calculate unknown vectors:
\begin{align}\label{iteration_algorithm}
\nonumber  & \mathbf{{c}}_{1}^{\left( m,q+1 \right)}=\mathbf{{c}}_{1}^{\left( m,q \right)}- \\
\nonumber &{{\left[ \frac{\partial \mathcal{L}\left( \mathbf{{c}}_{1}^{\left( m,q \right)},\bm{\kappa}_{1}^{\left( m,q \right)} \right)}{\partial \mathbf{{c}}_{1}^{\left( m \right)}\partial {{\left( \mathbf{{c}}_{1}^{\left( m \right)} \right)}^{T}}} \right]}^{-1}}\frac{\partial \mathcal{L}\left( \mathbf{{c}}_{1}^{\left( m,q \right)},\bm{\kappa}_{1}^{\left( m,q \right)} \right)}{\partial \mathbf{{c}}_{1}^{\left( m \right)}} \\
\nonumber & \bm{\kappa}_{1}^{\left( m,q+1 \right)}=\bm{\kappa}_{1}^{\left( m,q \right)}- \\
  &{{\left[ \frac{\partial \mathcal{L}\left( \mathbf{{c}}_{1}^{\left( m,q+1 \right)},\bm{\kappa}_{1}^{\left( m,q \right)} \right)}{\partial \bm{\kappa}_{1}^{\left( m \right)}\partial {{\left( \bm{\kappa}_{1}^{\left( m \right)} \right)}^{T}}} \right]}^{-1}}\frac{\partial \mathcal{L}\left( \mathbf{{c}}_{1}^{\left( m,q+1 \right)},\bm{\kappa}_{1}^{\left( m,q \right)} \right)}{\partial \bm{\kappa}_{1}^{\left( m \right)}},
\end{align}
where $q$ denotes the iteration index,
${\mathbf{{c}}_{1}^{\left( m \right)}},{\bm{\kappa}_{1}^{\left( m \right)}}\in {{\mathbb{R}}^{\left( K-1 \right)\times 1}}$ are defined in
$\mathbf{{c}}^{\left( m \right)}={{\left[ 1,\left(\mathbf{{c}}_{1}^{\left( m \right)}\right)^{T} \right]}^{T}}$ and
$\bm{\kappa}^{\left( m \right)}={{\left[ 1,\left(\bm{\kappa}_{1}^{\left( m \right)}\right)^{T} \right]}^{T}}$, respectively.
In particular, we set the initial values as
$\mathbf{{c}}_{1}^{\left( m,0 \right)}={{\mathbf{0}}_{\left( K-1 \right)\times 1}}$
and
$\bm{\kappa}_{1}^{\left( m,0 \right)}={{\mathbf{0}}_{\left( K-1 \right)\times 1}}$,
since the RTMO and RCFO errors are small after the compensation in (\ref{first_compensate}).
After convergence, we obtain the estimates $\mathbf{\hat{c}}^{\left( m \right)}={{\left[ 1,{{{\hat{c}}}_{1}^{\left( m \right)}},\cdots ,{{{\hat{c}}}_{K-1}^{\left( m \right)}} \right]}^{T}} \in \mathbb{C}^{K \times 1}$
and
$\bm{\hat{\kappa }}^{\left( m \right)}={{\left[ 1,{{{\hat{\kappa }}}_{1}^{\left( m \right)}},\cdots ,{{{\hat{\kappa }}}_{K-1}^{\left( m \right)}} \right]}^{T}} \in \mathbb{C}^{K \times 1}$.
The terms $\Delta \tilde{f}_{o,k}^{R{\left( m \right)}}$
and
$\Delta \tau _{o,k}^{R{\left( m \right)}}$ are related to
${{{\hat{c}}}_{k}^{\left( m \right)}}$ and ${{{\hat{\kappa }}}_{k}^{\left( m \right)}}$
via:
$\Delta \hat{\tilde{f}}_{o,k}^{R{\left( m \right)}}={{\hat{c}}_{k}^{\left( m \right)}}$ and $\Delta \hat{\tau }_{o,k}^{R{\left( m \right)}}=\frac{{{{\hat{\kappa }}}_{k}^{\left( m \right)}}}{2\pi \Delta f}$.
Similar to (\ref{final_coarse_RTMO_RCFO}), the refined estimates are obtained by the mean processing along $m$:
\begin{align}\label{final_RTMO_RCFO}
\Delta \hat{\tau }_{o,k}^{R}={\frac{1}{M}}\sum\limits_{m=0}^{M-1}{\Delta \hat{\tau }_{o,k}^{R\left( m \right)}}, \
 \Delta \hat{\tilde{f}}_{o,k}^{R}={\frac{1}{M}}\sum\limits_{m=0}^{M-1}{\Delta \hat{\tilde{f}}_{o,k}^{R\left( m \right)}},
\end{align}
where $M$ is the same as (\ref{final_coarse_RTMO_RCFO}).

The estimates in (\ref{final_RTMO_RCFO}) can then be used to compensate the residual errors in (\ref{first_compensate}):
\begin{align}\label{second_compensate}
 \nonumber \tilde{\tilde{y}}_{n,k}^{\left( m \right)}&={{e}^{-j\Delta \hat{\tilde{f}}_{o,k}^{R}}}{{e}^{j2\pi n\Delta \hat{\tau }_{o,k}^{R}\Delta f}}\tilde{y}_{n,k}^{\left( m \right)}\\
  &\approx \sum\limits_{\ell =0}^{L-1}{{{{\tilde{\xi }}}_{\ell }}{{e}^{-j2\pi n\left( {{\tau }_{\ell }}+{{\tau }_{o,0}} \right)\Delta f}}{{e}^{j2\pi k{{f}_{D,\ell }}{{T}_{0}}}}{a_{\ell}}{\left( m \right)}}.
\end{align}
Only ${\tau }_{o,0}$ is unknown. We summarize the major steps of the RTMO and RCFO estimation algorithm in Algorithm \ref{Algorithm1}.

\begin{remark}\label{Final_algorithm_advantage}
In this paper, the echo reflected from each target is regarded as a point source and there is only one echo for each object
in delay, Doppler, and angle dimension, as illustrated in (\ref{formula_original}).
In some practical applications, the micro-Doppler effect could be generated when a long coherent integration is performed \cite{9393557}.
We highlight that the proposed algorithm does not restrain the echo form and is also applicable to such challenging models with the micro-Doppler signatures.
Our RTMO and RCFO estimation algorithm works as long as A1 is satisfied.
For the scenarios of interest, there are typically less dynamic objects than static objects so that the proposed algorithm has a wide range of application prospects.
\end{remark}

\begin{algorithm}[t!]
\caption{Proposed RTMO and RCFO Estimation Algorithm}\label{Algorithm1}
\label{Algorithm1}
\hspace*{0.02in}{\bf Input:}
the received signal $y_{n,k}^{\left( m \right)}$, $m = 0, \cdots, M-1$, where $M \le {M_t}{M_r}$.\\
\hspace*{0.02in}{\bf Output:}
the signal after RTMO and RCFO compensation $\tilde{\tilde{y}}_{n,k}^{\left( m \right)}$
\begin{algorithmic}[1]
\FOR{$m=0$ to $M-1$}
\STATE Construct ${{\mathbf{y}}^{\left(m\right)}_{k}}$ using $y_{n,k}^{\left( m \right)}$. Calculate $\bm{\tilde{\varpi} }_{k,0}^{\left( m \right)}$ by (\ref{3}) and (\ref{4}), where $k=1,\cdots,K-1$.\\
\STATE Estimate the $\hat{\tau }_{o,k}^{R{\left( m \right)}}$ and $\hat{\tilde{f}}_{o,k}^{R\left( m \right)}$ by  (\ref{MLE_fomular_RTMO_RCFO}).
\ENDFOR
\STATE Collect all of the estimates and perform (\ref{final_coarse_RTMO_RCFO}) to calculate the coarse RTMO and RCFO estimates, i.e., $\hat{\tau }_{o,k}^{R}$ and $\hat{\tilde{f}}_{o,k}^{R}$.
Utilize $\hat{\tau }_{o,k}^{R}$ and $\hat{\tilde{f}}_{o,k}^{R}$ to compensate $y_{n,k}^{\left( m \right)}$ and obtain ${\tilde{y}}_{n,k}^{\left( m \right)}$.
\FOR{$m=0$ to $M-1$}
\STATE Set the initial values as $\mathbf{c}_{1}^{\left( m \right)}={{\mathbf{0}}_{ \left(K - 1\right) \times 1}}$ and $\bm{\kappa }_{1}^{\left( m\right)}={{\mathbf{0}}_{\left(K\right)\times 1}}$
and employ (\ref{iteration_algorithm}) to solve the optimization problem (\ref{9}).
\STATE Obtain $\Delta \tilde{f}_{o,k}^{R{\left( m \right)}}$
and
$\Delta \tau _{o,k}^{R{\left( m \right)}}$ after convergence.
\ENDFOR
\STATE Perform (\ref{final_RTMO_RCFO}) to calculate the refined RTMO and RCFO estimates, i.e., $\Delta \hat{\tau }_{o,k}^{R}$ and $\Delta \hat{\tilde{f}}_{o,k}^{R}$.\\
\STATE Compensate ${\tilde{y}}_{n,k}^{\left( m \right)}$ by $\Delta \hat{\tau }_{o,k}^{R}$ and $\Delta \hat{\tilde{f}}_{o,k}^{R}$ and obtain $\tilde{\tilde{y}}_{n,k}^{\left( m \right)}$.
\end{algorithmic}
\end{algorithm}

\section{Performance Analysis of Algorithm \ref{Algorithm1}}\label{Performance}
In this section, we characterize the estimation errors for RTMO and RCFO, based on received signals from one antenna and $M=1$ in Algorithm \ref{Algorithm1}.

\subsection{Derivation of RTMO and RCFO Estimation Errors}\label{Performance-A}
The central part of our algorithm is the iteration algorithm in Section \ref{III-B} since the calculation in Section \ref{III-A} is only performed for finding initial values.
In this subsection, we investigate the estimation performance of RTMO and RCFO for the proposed algorithm in Section \ref{III-B}.

It is worth noting that (\ref{7}) is used in Section \ref{III-B},
and there are three parts in $\mathbf{\tilde{y}}_{n}^{\left( m \right)}$: static, dynamic and noise part.
For ease of derivation, we redefine $\mathbf{\tilde{y}}_{n}^{\left( m \right)}$ based on A1:
\begin{align}\label{redefinition_y}
\mathbf{\tilde{y}}_{n}^{\left( m \right)}\approx\text{diag}\left\{ \Delta {{\mathbf{f}}^{R}}\odot \Delta \bm{\gamma }_{n}^{R} \right\}\mathbf{\tilde{y}}_{s,n}^{\left( m \right)}+\mathbf{z}_{n}^{\left( m \right)},
\end{align}
where $\mathbf{\tilde{y}}_{s,n}^{\left( m \right)} \in \mathbb{C}^{K \times 1}$ denotes the vector of the static part without the clock asynchronism.

Since $c_0$ and $\kappa _{0}$ are fixed in (\ref{9}),
we now define
$\mathbf{{z}}_{n}^{\prime\left( m \right)},\mathbf{{\tilde{y}}}_{s,n}^{\prime\left( m \right)}\in {{\mathbb{C}}^{\left( K-1 \right)\times 1}}$
from
$ \mathbf{\tilde{y}}_{s,n}^{\left( m \right)}={{\left[ {\tilde{y}}_{s,n,0}^{\left( m \right)},\left( \mathbf{{\tilde{y}}}_{s,n}^{\prime\left( m \right)} \right)^{T} \right]}^{T}}$,
$\mathbf{{z}}_{n}^{\left( m \right)}={{\left[ {{z}}_{n,0}^{\left( m \right)},\left(  \mathbf{{{z}}}_{n}^{\prime\left( m \right)} \right)^{T} \right]}^{T}}$
to assist the derivation of estimation errors.
Note that $\mathbf{\tilde{y}}_{s,n}^{\left( m \right)}={{\mathbf{1}}_{K\times 1}}h_{s,n}^{\left( m \right)}$ based on (\ref{two_parts}).
We have the following proposition.
\begin{proposition}\label{propo_matrix}
The error of the vectors $\mathbf{c}_{1}^{\left( m \right)} \in \mathbb{R}^{\left(K - 1\right)\times 1}$ and $\bm{\kappa}_{1}^{\left( m \right)}\in \mathbb{R}^{\left(K - 1\right)\times 1}$ estimated by (\ref{iteration_algorithm}) is:
\begin{align}\label{propos_formula}
\nonumber \Delta \mathbf{c}_{1}^{\left( m \right)} & \approx {{\left( \sum\limits_{n=0}^{N-1}{{{\left| h_{s,n}^{\left( m \right)} \right|}^{2}}} \right)}^{-1}}\sum\limits_{n=0}^{N-1}{\operatorname{Im}\left\{ \bm{\mu }_{n}^{\left( m \right)} \right\}}, \\
\Delta \bm{\kappa }_{1}^{\left( m \right)} & \approx -{{\left( \sum\limits_{n=0}^{N-1}{{{n}^{2}}{{\left| h_{s,n}^{\left( m \right)} \right|}^{2}}} \right)}^{-1}}\sum\limits_{n=0}^{N-1}{n\operatorname{Im}\left\{ \bm{\mu }_{n}^{\left( m \right)} \right\}},
\end{align}
where $\bm{\mu }_{n}^{\left( m \right)}={{\left( h_{s,n}^{\left( m \right)} \right)}^{*}} \mathbf{{{z}}}_{n}^{\prime\left( m \right)}  \in \mathbb{C}^{\left( K - 1 \right) \times 1}$,
and $h_{s,n}^{\left( m \right)}$ is the value of the static part defined in (\ref{approximation_formula_11}).
\end{proposition}
\begin{proof}
Please refer to Appendix \ref{Appendix_A}.
\end{proof}
It is worth noting that Algorithm \ref{Algorithm1} generate biased estimates in the presence of dynamic targets.
Fortunately, the denominator of (\ref{propos_formula}) can be regarded as the static objects' power, i.e.,
${\mathcal{P}_{s}}=\frac{1}{N}\sum\nolimits_{n=0}^{N-1}{{{\left| h_{s,n}^{\left( m \right)} \right|}^{2}}}$,
so the bias is small based on A1.
Our algorithm can be considered to be approximately unbiased.

\begin{remark}\label{independent_of_RTMO_RCFO}
 \emph{Estimation errors are irrelevant to RTMO and RCFO}:
Note that $ \mathbf{{z}}_{n}^{\left( m \right)}$ and $h_{s,n}^{\left( m \right)}$ in (\ref{propos_formula}) are both independent of RTMO (${{\tau }_{o,0}^R}$) and RCFO (${\tilde{f }_{o,0}^R}$).
This indicates that the proposed algorithm will maintain the same estimation errors $\Delta \mathbf{{c}}_{1}^{\left( m \right)}$ and $\Delta \bm{{\kappa}}_{1}^{\left( m \right)}$ for any values of RTMO and RCFO,
which shows the stability of the proposed algorithm to clock asynchronism.
\end{remark}

The estimation variance of RTMO and RCFO is given in Proposition \ref{propo_var}.
\begin{proposition}\label{propo_var}
For complex noise with equal variance of real and imaginary parts, i.e.,
$\operatorname{var}\left\{ \operatorname{Im}\left\{ \mathbf{z}_{n}^{\left( m \right)} \right\} \right\}=\operatorname{var}\left\{ \operatorname{Re}\left\{ \mathbf{z}_{n}^{\left( m \right)} \right\} \right\} = {{\mathbf{1}}_{K\times 1}}{\sigma _{n}^{2}}/{2}\;$.
The estimation error variance is:
\begin{align}\label{propo_varvar}
\nonumber \operatorname{var}\left\{ \Delta \mathbf{c}_{1}^{\left( m \right)} \right\} & \approx \sigma _{n}^{2}{{\left( \sum\limits_{n=0}^{N-1}{{{\left| h_{s,n}^{\left( m \right)} \right|}^{2}}} \right)}^{-1}}{{\mathbf{1}}_{\left( K-1 \right)\times 1}}, \\
\operatorname{var}\left\{ \Delta \bm{\kappa }_{1}^{\left( m \right)} \right\} & \approx \sigma _{n}^{2}{{\left( \sum\limits_{n=0}^{N-1}{{{n}^{2}}{{\left| h_{s,n}^{\left( m \right)} \right|}^{2}}} \right)}^{-1}}{{\mathbf{1}}_{\left( K-1 \right)\times 1}}.
\end{align}
\end{proposition}
\begin{proof}
Please refer to Appendix \ref{Appendix_AABB}.
\end{proof}
Note that $\sum\nolimits_{n=0}^{N-1}{{{n}^{2}}{{\left| h_{s,n}^{\left( m \right)} \right|}^{2}}}>\sum\nolimits_{n=0}^{N-1}{{{\left| h_{s,n}^{\left( m \right)} \right|}^{2}}}$ in (\ref{propos_formula}) and (\ref{propo_varvar}).
The estimation for RTMO is generally more accurate than that for RCFO.
Moreover, the RTMO accuracy is more sensitive to the value of $N$.

\subsection{Influence of RTMO and RCFO Errors on Target Estimation}\label{Performance-B}

In this subsection, we use an example of FFT-based method \cite{7833233} to estimate the sensing parameters and compare the sensing difference between the compensated signal $\tilde{\tilde{y}}_{n,k}^{\left( m \right)}$ in (\ref{second_compensate}) and the signal without clock asynchronism.
For the noise term,
it is worth noting that
the compensation in (\ref{first_compensate}) and (\ref{second_compensate}) only changes the noise phase,
so the noise power in $\tilde{\tilde{y}}_{n,k}^{\left( m \right)}$ remains the same as the power of the original noise ${{z}_{n,k}^{\left(m\right)}}$.
Therefore, we only need to consider the impact of RTMO and RCFO errors on the non-noise term:
\begin{align}\label{formula_fft_influence}
 \nonumber \Delta \mathcal{F} & \approx {{\mathcal{F}}_{2}}\left\{ {{e}^{j\Delta {{c}_{k}}-jn\Delta {{\kappa }_{k}}}}\tilde{y}_{s,n,k}^{\left( m \right)} \right\}-{{\mathcal{F}}_{2}}\left\{ \tilde{y}_{s,n,k}^{\left( m \right)} \right\} \\
 \nonumber  & =\frac{1}{NK}{{\mathcal{F}}_{2}}\left\{ {{e}^{j\Delta {{c}_{k}}-jn\Delta {{\kappa }_{k}}}}-1 \right\}\circledast {{\mathcal{F}}_{2}}\left\{ \tilde{y}_{s,n,k}^{\left( m \right)} \right\} \\
 & \triangleq {{\mathcal{E}}_{\Delta }}\circledast {{\mathcal{F}}_{2}}\left\{ \tilde{y}_{s,n,k}^{\left( m \right)} \right\},
\end{align}
where $\tilde{y}_{s,n,k}^{\left( m \right)}$ is defined from (\ref{redefinition_y})
with ${{\left[ \mathbf{\tilde{y}}_{s,n}^{\left( m \right)} \right]}_{k}}=\tilde{y}_{s,n,k}^{\left( m \right)}$.
For the influence of ${{\mathcal{E}}_{\Delta }}\circledast {{\mathcal{F}}_{2}}\left\{ \tilde{y}_{s,n,k}^{\left( m \right)} \right\}$, we have the following proposition.
\begin{proposition}\label{propo_target_fft}
The influence of ${{\mathcal{E}}_{\Delta }}\circledast {{\mathcal{F}}_{2}}\left\{ \tilde{y}_{s,n,k}^{\left( m \right)} \right\}$ can be neglected since its power is much lower than the noise power.
\end{proposition}
\begin{proof}
Please refer to Appendix \ref{Appendix_B}.
\end{proof}

To sum up, (\ref{second_compensate}) can be directly used for target estimation since the RTMO and RCFO errors have ignorable influence in delay and Doppler.
Moreover, the errors do not influence AOD and AOA estimation since the compensation in (\ref{second_compensate}) is identical for the signal in different antennas.

\section{Localization of UE and Dynamic Targets}\label{IV}
In this section,
we present the proposed localization scheme for UE and dynamic targets
assisted by the anchor points in the absence of the LOS path between UE and BS.

With the RTMO and RCFO estimated and compensated, conventional sensing techniques can be applied to estimate the parameters AOA (for dynamic targets only), AOD (for anchor points only), delay, and Doppler, almost like in a system without clock offset, except that the delay is relative to $\tau_{o,0}$.
We first briefly describe how to estimate these parameters for dynamic targets and anchor points, and then elaborate how to
estimate $\tau_{o,0}$ by exploiting the anchor points.

\subsection{Parameter Estimation}\label{IV-A}

For anchor points,we exploit the prior knowledge as assumed in A2.
We first extract signals with $f_D = 0$, from which we construct the range-AOA spectrum initially.
Since the AOA of each anchor point is precisely known according to A2,
we extract the range spectrum based on the known AOA and
record the \emph{relative range} value (note that the range value is contaminated by $\tau_{o, 0}$) corresponding to the highest energy point in each range spectrum.
The recorded data belongs to the anchor point if no extra static objects exist between the BS and the anchor point.
Although interference may appear and impact the recognition of anchor points,
we highlight that the interference is generally not constant, and the true anchor points can be effectively identified with high probability by multiple observations.
In what follows, we do not consider the recognition failure case of anchor points.

For dynamic target estimation, the static objects, regarded as interference, are able to be eliminated by Doppler filtering, such as the moving targets detection (MTD) method \cite{RadarHandbook}.
Then, we perform the typical sensing algorithm to estimate the dynamic parameters,
such as the MUSIC, ESPRIT, compressed sensing, FFT-based method, or tensor-based algorithms \cite{7833233, 9585321}.
Note that the estimated range is also contaminated by $\tau_{o, 0}$.

We now define the estimated parameters as follows:
\begin{itemize}
  \item
  Estimates of dynamic targets: $ {{\bm{\hat{\theta} }}_{d}}={{\left[ {\hat{\theta} }_{0}^{\left( d \right)},\cdots ,{\hat{\theta} }_{{{L}_{d}}-1}^{\left( d \right)} \right]}^{T}} $
  and
  $ {{\hat{\tilde{\mathbf{{R}}}}}_{d}}={{\left[ \hat{\tilde{R}}_{0}^{\left( d \right)},\cdots ,\hat{\tilde{R}}_{{{L}_{d}}-1}^{\left( d \right)} \right]}^{T}} $
  denote the AOA and the relative range of dynamic targets, respectively;
  \item
  Estimates of anchor points: $ {\hat{\bm{\varphi }}_{a}}={{\left[ \hat{\varphi} _{0}^{\left( a \right)},\cdots ,\hat{\varphi} _{{{L}_{a}}-1}^{\left( a \right)} \right]}^{T}}$
  and
  $ {{\hat{\tilde{\mathbf{{R}}}}}_{a}} = {{\left[ \hat{\tilde{R}}_{0}^{\left( a \right)},\cdots ,\hat{\tilde{R}}_{{{L}_{a}}-1}^{\left( a \right)} \right]}^{T}} $
  denote the AOD and the relative range of anchor points, respectively.
\end{itemize}
where $\tilde{R}_{{{\ell }_{d}}}^{\left( d \right)}$ and $\tilde{R}_{{{\ell }_{a}}}^{\left( a \right)}$
denote the bistatic range contaminated by RTMO, i.e.,
$\tilde{R}_{{{\ell }_{d}}}^{\left( d \right)}=c\left( {{\tau }_{o,0}}+{{\tau }_{{{\ell }_{d}}}} \right)$
and
$ \tilde{R}_{{{\ell }_{a}}}^{\left( a \right)}=c\left( {{\tau }_{o,0}}+{{\tau }_{{{\ell }_{a}}}} \right)$.
It is noted that we do not need to estimate the AOD of the dynamic targets for localization,
and the AOA of anchor points is known as a priori.

\subsection{Localization Algorithm}\label{IV-B}
\begin{figure}[!t]
\centering
\includegraphics[width=3.2in]{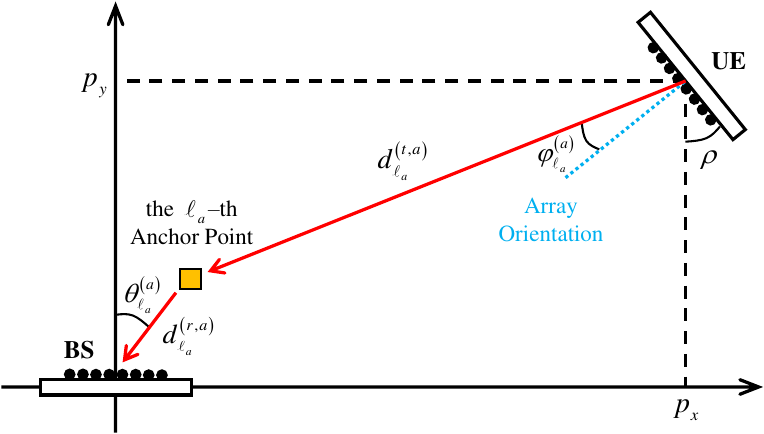}
\caption{Illustration of the parameters used for UE localization.}
\label{Fig2}
\end{figure}

In this subsection, we utilize the parameters of $L_{a}$ anchor points obtained in Section \ref{IV-A} to locate the
UE position and estimate $\tau_{o, 0}$ using the anchor point parameters,
 and subsequently locate each dynamic target.
We underline that the array orientation of UE is unknown to the BS.
As depicted in Fig. \ref{Fig2}, we use $\rho$ to describe the angle between the UE array and $y$-axis.

Let
$\left(p_x, p_y\right)$,
$d_{{{\ell }_{a}}}^{\left( t,a \right)}$
and
$d_{{{\ell }_{a}}}^{\left( r,a \right)}$
denote
the cartesian coordinates of the UE,
the range between the UE and the $\ell_a$-th anchor point,
and the range between the $\ell_a$-th anchor point and the BS,
respectively,
as shown in Fig. \ref{Fig2}.
The absolute bistatic range of the $\ell_a$-th anchor point can be expressed as
$c{{\tau }_{{{\ell }_{a}}}}=d_{{{\ell }_{a}}}^{\left( t,a \right)}+d_{{{\ell }_{a}}}^{\left( r,a \right)}$,
so
$\tilde{R}_{{{\ell }_{a}}}^{\left( a \right)}=c{{\tau }_{o,0}}+c{{\tau }_{{{\ell }_{a}}}}=c{{\tau }_{o,0}}+d_{{{\ell }_{a}}}^{\left( t,a \right)}+d_{{{\ell }_{a}}}^{\left( r,a \right)}$.
We can also further derive the following relationships using the parameters of anchor points:
\begin{subequations}\label{geometric}
\begin{align}
  & {\tilde{R}}_{{{\ell }_{a}}}^{\left( a \right)}-{\tilde{R}}_{0}^{\left( a \right)}= \left( d_{{{\ell }_{a}}}^{\left( t,a \right)}+d_{{{\ell }_{a}}}^{\left( r,a \right)} \right)-\left( d_{0}^{\left( t,a \right)}+d_{0}^{\left( r,a \right)} \right), \label{geometric_R_R}\\
 & {{p}_{x}}=\cos \left( \rho -{\varphi }_{{{\ell }_{a}}}^{\left( a \right)} \right)d_{{{\ell }_{a}}}^{\left( t,a \right)}+\sin \left( \theta _{{{\ell }_{a}}}^{\left( a \right)} \right)d_{{{\ell }_{a}}}^{\left( r,a \right)}, \label{geometric_px}\\
&{{p}_{y}}=\sin \left( \rho -{\varphi }_{{{\ell }_{a}}}^{\left( a \right)} \right)d_{{{\ell }_{a}}}^{\left( t,a \right)}+\cos \left( \theta _{{{\ell }_{a}}}^{\left( a \right)} \right)d_{{{\ell }_{a}}}^{\left( r,a \right)},\label{geometric_py}
\end{align}
\end{subequations}
where
$\theta _{{{\ell }_{a}}}^{\left( a \right)}$ denotes the AOA of the $\ell_a$-th anchor point,
and $d_{{{\ell }_{a}}}^{\left( r,a \right)}$ and $\theta _{{{\ell }_{a}}}^{\left( a \right)}$ are precisely known based on A2.
The index in (\ref{geometric_R_R}) is $\ell_a = 1, \cdots, L_{a} - 1$, whereas the index in (\ref{geometric_px}) and (\ref{geometric_py}) is $\ell_a = 0, \cdots, L_{a} - 1$.
It is worth noting that $\tau_{o, 0}$ is effectively eliminated in (\ref{geometric_R_R}).
We now define three types of variables based on (\ref{geometric}):
\begin{itemize}
  \item
  $g_{{{\ell }_{a}}}^{\left( R \right)}  \triangleq  \left( \tilde{R}_{{{\ell }_{a}}}^{\left( a \right)}-\tilde{R}_{0}^{\left( a \right)} \right)-\left[ \left( d_{{{\ell }_{a}}}^{\left( t,a \right)}+d_{{{\ell }_{a}}}^{\left( r,a \right)} \right)-\left( d_{0}^{\left( t,a \right)}+d_{0}^{\left( r,a \right)} \right) \right]$,
  $\ell_a = 1, \cdots, L_a - 1$;
  \item
  $g_{{{\ell }_{a}}}^{\left( {{p}_{x}} \right)}  \triangleq  {{p}_{x}}-\left[ \cos \left( \rho -\varphi _{{{\ell }_{a}}}^{\left( a \right)} \right)d_{{{\ell }_{a}}}^{\left( t,a \right)}+\sin \left( \theta _{{{\ell }_{a}}}^{\left( a \right)} \right)d_{{{\ell }_{a}}}^{\left( r,a \right)} \right]$,
  $\ell_a = 0, \cdots, L_a - 1$;
  \item
  $g_{{{\ell }_{a}}}^{\left( {{p}_{y}} \right)}  \triangleq  {{p}_{y}}-\left[ \sin \left( \rho -\varphi _{{{\ell }_{a}}}^{\left( a \right)} \right)d_{{{\ell }_{a}}}^{\left( t,a \right)}+\cos \left( \theta _{{{\ell }_{a}}}^{\left( a \right)} \right)d_{{{\ell }_{a}}}^{\left( r,a \right)} \right]$,
  $\ell_a = 0, \cdots, L_a - 1$.
\end{itemize}
We can thus obtain $3L_a - 1$ equations.
Stacking all unknown parameters into a vector
$\bm{\varsigma }=  \  \left[ {{p}_{x}},{{p}_{y}},\rho ,d_{0}^{\left( t,a \right)},\cdots ,d_{{{L}_{a}}-1}^{\left( t,a \right)} \right]^{T}\in {{\mathbb{R}}^{\left( {{L}_{a}}+3 \right)\times 1}}$, we have the following relationship:
\begin{align}\label{gxxxxxggg}
\nonumber \mathbf{g}\left( \bm{\varsigma } \right)&={{\left[ g_{1}^{\left( R \right)},\cdots ,g_{{{L}_{a}}-1}^{\left( R \right)},g_{0}^{\left( {{p}_{x}} \right)},\cdots ,g_{{{L}_{a}}-1}^{\left( {{p}_{x}} \right)},g_{0}^{\left( {{p}_{y}} \right)},\cdots ,g_{{{L}_{a}}-1}^{\left( {{p}_{y}} \right)} \right]}^{T}}\\
&={{\mathbf{0}}_{\left( 3{{L}_{a}}-1 \right)\times 1}},
\end{align}
where $\mathbf{g}\left( {{\bm{\varsigma }}} \right)\in \mathbb{R}^{\left(3L_a - 1\right)\times1}$.
This problem can be solved if the number of unknown parameters is less than the number of equations, i.e., $L_a + 3 \le 3L_a - 1$. We obtain:
\begin{align}\label{the_lowest_requirement_of_sensing}
{{L}_{a}}\ge 2.
\end{align}

Estimating $\bm{\varsigma }$ from (\ref{gxxxxxggg}) is a nonlinear problem and
there is no closed-form solution.
We exploit the non-linear least squares method to get the estimates and locate the UE.
The method is based on the iteration algorithm:
\begin{align}\label{localization_iteration}
{{\bm{\varsigma }}^{\left( q+1 \right)}}={{\bm{\varsigma }}^{\left( q \right)}}-{{\mathbf{H}}^{\dagger}}\left( {{\bm{\varsigma }}^{\left( q \right)}} \right)\mathbf{g}\left( {{\bm{\varsigma }}^{\left( q \right)}} \right),
\end{align}
where $q$ represents the iteration index;
$ \mathbf{H}\left( \bm{\varsigma } \right)\in {{\mathbb{R}}^{\left( 3{{L}_{a}}-1 \right)\times \left( {{L}_{a}}+3 \right)}}$ denotes the Jacobi matrix of $\mathbf{g}\left( {{\bm{\varsigma }}^{\left( q \right)}} \right)$ given in (\ref{Localizaiton H matrix}) at the top of the next page.
\begin{figure*}[ht]
\begin{align}\label{Localizaiton H matrix}
\mathbf{H}\left( \bm{\varsigma } \right)=\left[ \begin{matrix}
   {{\mathbf{0}}_{\left( {{L}_{a}}-1 \right)\times 2}} & {{\mathbf{0}}_{\left( {{L}_{a}}-1 \right)\times 1}} & \left[ -{{\mathbf{1}}_{\left( {{L}_{a}}-1 \right)\times 1}},{{\mathbf{I}}_{{{L}_{a}}-1}} \right]  \\
   -{{\mathbf{1}}_{{{L}_{a}}\times 1}}\otimes {{\mathbf{I}}_{2}} & \text{diag}\left\{ {{\bm{\varepsilon }}^{\left( t,a \right)}} \right\}\left( {{\mathbf{d}}^{\left( t,a \right)}}\otimes {{\left[ -1,1 \right]}^{T}} \right) & \left( {{\mathbf{I}}_{{{L}_{a}}}}\otimes {{\mathbf{T}}_{2}} \right)\text{diag}\left\{ {{\bm{\varepsilon }}^{\left( t,a \right)}} \right\}\left( {{\mathbf{I}}_{{{L}_{a}}}}\otimes {{\mathbf{1}}_{2\times 1}} \right)  \\
\end{matrix} \right].
\end{align}
\normalsize
\hrulefill
\vspace*{4pt}
\end{figure*}
In (\ref{Localizaiton H matrix}),
${{\mathbf{d}}^{\left( t,a \right)}}={{\left[ d_{0}^{\left( t,a \right)},\cdots ,d_{{{L}_{a}}-1}^{\left( t,a \right)} \right]}^{T}}\in {{\mathbb{C}}^{{{L}_{a}}\times 1}}$,
${{\mathbf{T}}_{2}}=\left[ \begin{matrix}
   0 & 1  \\
   1 & 0  \\
\end{matrix} \right] $, and
${{\bm{\varepsilon }}^{\left( t,a \right)}}=\left[ \sin \left( \rho -\varphi _{0}^{\left( a \right)} \right),\cos \left( \rho -\varphi _{0}^{\left( a \right)} \right),\cdots,\sin \left( \rho -\varphi _{{{L}_{a}}-1}^{\left( a \right)} \right)  \right.$
$  {{\left.,\cos \left( \rho -\varphi _{{{L}_{a}}-1}^{\left( a \right)} \right) \right]}^{T}}\in {{\mathbb{C}}^{2{{L}_{a}}\times 1}}$.

The absolute location of UE is thus estimated as:
\begin{align}\label{UE and dynamic target position}
\left( {{{\hat{p}}}_{x}},{{{\hat{p}}}_{y}} \right)=\left( {{\left[ {\hat{\bm{\varsigma }}} \right]}_{0}},{{\left[ {\hat{\bm{\varsigma }}} \right]}_{1}} \right).
\end{align}
Then, we estimate the reference TMO, i.e., $\tau_{o, 0}$, based on the relationship
$\tilde{R}_{{{\ell }_{a}}}^{\left( a \right)}=c{{\tau }_{o,0}}+d_{{{\ell }_{a}}}^{\left( t,a \right)}+d_{{{\ell }_{a}}}^{\left( r,a \right)}$ given above (\ref{geometric}):
\begin{align}\label{reference TMO estimate}
 \nonumber {\hat{\tau }_{o,0}}&=\frac{1}{c L_a}\sum\limits_{{{\ell }_{a}}=0}^{{{L}_{a}}-1}{\left[ \hat{\tilde{R}}_{{{\ell }_{a}}}^{\left( a \right)} \right.} \\
 & \left. -\left( {{\left\| \left( {{{\hat{p}}}_{x}},{{{\hat{p}}}_{y}} \right)-\left( p_{x,{{\ell }_{a}}}^{\left( \text{ANP} \right)},{p}_{y,{{\ell }_{a}}}^{\left( \text{ANP} \right)} \right) \right\|}_{2}}+d_{{{\ell }_{a}}}^{\left( r,a \right)} \right) \right],
\end{align}
where $\left( p_{x,{{\ell }_{a}}}^{\left( \text{ANP} \right)},{p}_{y,{{\ell }_{a}}}^{\left( \text{ANP} \right)} \right)$ represents the cartesian coordinates of the $\ell_a$-th anchor point.

Finally, we are ready to estimate the absolute locations for the dynamic targets.
According to (\ref{reference TMO estimate}), the absolute delay of the $\ell_d$-th dynamic target is
${\hat{\tau }_{{{\ell }_{d}}}}=\frac{1}{c}\hat{\tilde{R}}_{{{\ell }_{d}}}^{\left( d \right)}-{{\hat{\tau }}_{o,0}}$.
Let $d_{{{\ell }_{d}}}^{\left( t,d \right)}$ and $d_{{{\ell }_{d}}}^{\left( r,d \right)}$ denote
the range between UE and dynamic target and the range between dynamic target and BS, respectively.
We obtain $c{{\tau }_{{{\ell }_{d}}}} = d_{{{\ell }_{d}}}^{\left( t,d \right)} + d_{{{\ell }_{d}}}^{\left( r,d \right)}$.
Using ${\hat{\tau }_{{{\ell }_{d}}}}$, $\left( {{{\hat{p}}}_{x}},{{{\hat{p}}}_{y}} \right)$ and $\hat{\theta}_{\ell_d}^{\left(d\right)}$,
the following equation with respect to the $\ell_d$-th dynamic target holds:
\begin{align}\label{equations for dynamic target localization}
c{{{\hat{\tau }}}_{{{\ell }_{d}}}}={{\left\| \left( d_{{{\ell }_{d}}}^{\left( r,d \right)}\sin \hat{\theta }_{{{\ell }_{d}}}^{\left( d \right)},d_{{{\ell }_{d}}}^{\left( r,d \right)}\cos \hat{\theta }_{{{\ell }_{d}}}^{\left( d \right)} \right)-\left( {{{\hat{p}}}_{x}},{{{\hat{p}}}_{y}} \right) \right\|}_{2}} +d_{{{\ell }_{d}}}^{\left( r,d \right)}.
\end{align}
We can obtain $\hat{d}_{{{\ell }_{d}}}^{\left( r,d \right)}$ from (\ref{equations for dynamic target localization}).
The absolute location of the $\ell_d$-th dynamic target is:
\begin{align}\label{dynamic location}
\left( \hat{x}_{{{\ell }_{d}}}^{\left( d \right)},\hat{y}_{{{\ell }_{d}}}^{\left( d \right)} \right)=\left( \hat{d}_{{{\ell }_{d}}}^{\left( r,d \right)}\sin \hat{\theta }_{{{\ell }_{d}}}^{\left( d \right)},\hat{d}_{{{\ell }_{d}}}^{\left( r,d \right)}\cos \hat{\theta }_{{{\ell }_{d}}}^{\left( d \right)} \right).
\end{align}
We summarize the major steps of the UE and dynamic target localization algorithm in Algorithm \ref{Algorithm2}.
Note that $g\left(\bm{\varsigma}\right)$ equals zero when $\bm{\varsigma}$ is the true value in a noise-free environment. Hence, our method is unbiased.

\begin{algorithm}[t!]
\caption{Proposed UE and Dynamic Target Localization Algorithm}\label{Algorithm2}
\hspace*{0.02in}{\bf Input:}
The estimates: $\hat{\tilde{R}}_{{{\ell }_{d}}}^{\left( d \right)}$, $\hat{\theta }_{{{\ell }_{d}}}^{\left( d \right)}$, $\hat{\tilde{R}}_{{{\ell }_{a}}}^{\left( a \right)}$ and $\hat{\varphi }_{{{\ell }_{a}}}^{\left( a \right)}$,
where $\ell_d = 0, \cdots, L_d - 1$ and $\ell_a = 0, \cdots, L_a - 1$.
The known position of anchor points: $\left( p_{x,{{\ell }_{a}}}^{\left( \text{ANP} \right)},{p}_{y,{{\ell }_{a}}}^{\left( \text{ANP} \right)} \right)$.\\
\hspace*{0.02in}{\bf Output:}
$\left( {{{\hat{p}}}_{x}},{{{\hat{p}}}_{y}} \right)$,
$\left( \hat{x}_{{{\ell }_{d}}}^{\left( d \right)},\hat{y}_{{{\ell }_{d}}}^{\left( d \right)} \right)$,
where $\ell_d = 0, \cdots, L_d - 1$.

\begin{algorithmic}[1]
\STATE Initialization: ${{\bm{\varsigma }}^{\left( 0 \right)}}={{\left[ 0,0,0,\hat{\tilde{R}}_{0}^{\left( a \right)},\cdots ,\hat{\tilde{R}}_{{{L}_{a}}-1}^{\left( a \right)} \right]}^{T}}$.
\WHILE {Unconvergence}
\STATE Update $\mathbf{g}\left( {{\bm{\varsigma }}^{\left( q \right)}} \right)$ and the Jacobi matrix  $\mathbf{H}\left( \bm{\varsigma }^{\left( q \right)} \right)$.
\STATE Calculate ${{\bm{\varsigma }}^{\left( q+1 \right)}}$ by (\ref{localization_iteration}). $q = q + 1$.
\ENDWHILE
\STATE Obtain $\left( {{{\hat{p}}}_{x}},{{{\hat{p}}}_{y}} \right)$ by (\ref{UE and dynamic target position}).
\STATE Calculate $\hat{d}_{{{\ell }_{d}}}^{\left( r,d \right)}$ by (\ref{reference TMO estimate}) (\ref{equations for dynamic target localization}).
Obtain $\left( \hat{x}_{{{\ell }_{d}}}^{\left( d \right)},\hat{y}_{{{\ell }_{d}}}^{\left( d \right)} \right)$ by (\ref{dynamic location}).
\end{algorithmic}
\end{algorithm}

\subsection{Feasibility of UE Localization}\label{IV-C}

In this subsection, we evaluate the impact of noise on the localization.
Similar to the analysis in (\ref{Appen_formu_2}), the UE localization accuracy is described by first-order approximation:
\begin{align}\label{1st_in_vc}
\Delta \bm{\varsigma }\approx \frac{\partial \bm{\varsigma }}{\partial \mathbf{e}_\text{all}^{T}}\Delta {{\mathbf{e}}_\text{all}}=-{{\mathbf{H}}^{\dagger }}\left( \bm{\varsigma } \right)\frac{\partial \mathbf{g}\left( \bm{\varsigma } \right)}{\partial \mathbf{e}_\text{all}^{T}}\Delta {{\mathbf{e}}_\text{all}},
\end{align}
where
${{\mathbf{e}}_\text{all}}={{\left[ \mathbf{R}_{a}^{T},\bm{\varphi }_{a}^{T} \right]}^{T}}\in {{\mathbb{R}}^{ 2{{L}_{a}} \times 1}}$
denotes the parameters used for localization;
$\Delta {{\mathbf{e}}_\text{all}}$ and $\Delta \bm{\varsigma } \in {{\mathbb{R}}^{\left({{L}_{a}}+3 \right)\times 1}}$ represent the estimation error vector and the localization error vector, respectively.
We provide the following proposition to illustrate the situation that needs to be averted.

\begin{proposition}\label{propo_angle_and_something}
The localization errors for the UE are infinite if the UE and all anchor points are on the same line.
\end{proposition}
\begin{proof}
When UE and all anchor points are on the same line, the AOD of anchor points is equal,
i.e.,
$\varphi \triangleq \varphi_{0}^{\left(a\right)} = \cdots = \varphi_{L_a - 1}^{\left(a\right)}$.
In this case, however, it is easy to verify that $ \left| {{\mathbf{H}}^{T}}\left( \bm{\varsigma } \right)\mathbf{H}\left( \bm{\varsigma } \right) \right|=0$,
which means ${{\mathbf{H}}^{\dagger }}$ is infinite.
Thus, the proposed algorithm fails while $\varphi_{0}^{\left(a\right)} = \cdots = \varphi_{L_a - 1}^{\left(a\right)}$.
\end{proof}

Essentially, the data is correlated when UE and all anchor points are on the same line, reducing the number of non-correlated equations, so our algorithm fails to locate the UE.
The solution to this problem is to employ more anchor points and avoid any three anchor points collinearly.

\section{Simulation Results}\label{V}
In this section, simulation results are presented to validate the proposed algorithms.
We first describe the system setup, then evaluate the sensing performance.

\subsection{System Setting}\label{V-A}
The carrier frequency is set to $3\text{GHz}$, the wavelength is $\lambda = 0.1\text{m}$, the subcarrier spacing of OFDM signal is $\Delta f = 480 \text{kHz}$, and the number of subcarrier is $N = 256$.
Therefore, the bandwidth is $B = N \Delta f = 122.88\text{MHz}$, the maximum unambiguous range is ${{R}_{\max }}={c}/{\Delta f}\;=625\text{m}$.
The interval between OFDM symbols is ${{T}_{0}}=62.5\text{ }\!\!\mu\!\!\text{ s}$,
and there are $K={10\text{ms}}/{{{T}_{0}}}\;=160$ OFDM symbols used for sensing.
The noise power is $\sigma_n^{2} = k_B F_n T_{st} B = 4.92\times10^{-12}\text{W}$, where $k_B$ is the Boltamann's constant, $F_n = 10$ is the receiver noise figure, and $T_{st} = 290\text{K}$ is the standard temperature.

The sizes of antenna arrays of BS and UE are $M_r = 8$ and $M_t = 4$, respectively, and the antenna interval is set to be half-wavelength.
Suppose that the positions of the static BS and UE are $\left( 0\text{m}, 0\text{m} \right)$ and $\left( 60\text{m}, 40\text{m} \right)$, respectively.
The angle between the UE array and $y$-axis, as depicted in Fig. \ref{Fig2}, is $\rho =20{}^\circ $.
\begin{figure}[!t]
\centering
\includegraphics[width=3.2in]{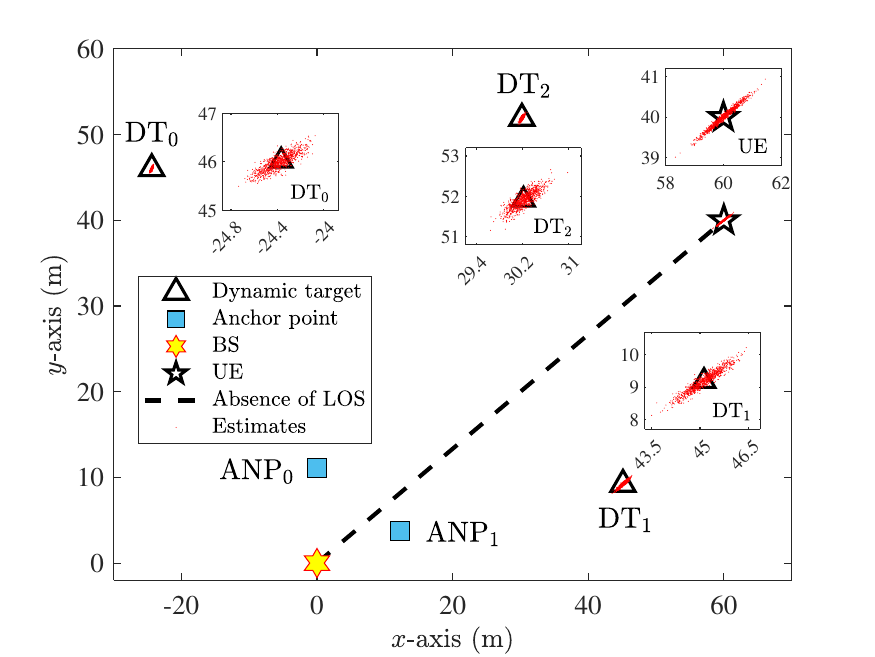}
\caption{Position of dynamic targets, BS, UE, and the localization results.}
\label{Fig_simu}
\end{figure}
Moreover, three dynamic targets are distributed in the space, and the positions are $\text{DT}_0 = \left( -24.4\text{m}, 46.0\text{m} \right)$, $\text{DT}_1 = \left( 45.1\text{m}, 9.2\text{m} \right)$, and $\text{DT}_2 = \left( 30.2\text{m}, 51.9\text{m} \right)$.
The complex-valued reflection coefficient is set based on (\ref{formula_original}).
Assume that the reflecting factor for dynamic target and the transmitted power are $\sigma_{x} = 1\text{m}^2$ and $0.1\text{W}$, respectively.

In addition, $L_s = 15$ static objects are assumed to be randomly distributed in space.
Based on the definition in Section \ref{Performance}, the averaged signal-to-noise ratio for static objects is defined as:
\begin{align}\label{SNR_for_static_target}
\text{SN}{{\text{R}}_{S}}=\frac{{\mathcal{P}_{s}}}{\sigma _{n}^{2}}=\frac{\sum\limits_{n=0}^{N-1}{{{\left| h_{s,n}^{\left( m \right)} \right|}^{2}}}}{N\sigma _{n}^{2}}.
\end{align}
Moreover, two anchor points are considered in this section, with positions of $\text{ANP}_0 = \left( 0\text{m}, 11.1\text{m} \right)$ and $\text{ANP}_1 = \left( 12.3\text{m}, 3.8\text{m} \right)$, as shown in Fig. \ref{Fig_simu}.
Assume that the reflecting factor for each anchor point is $\sigma_{x} = 2\text{m}^2$ \cite{RadarHandbook}.
In this case, the signal power of the static objects is 11dB higher than that of the dynamic targets.

We use large CFO and TMO values in this section to verify Remark \ref{independent_of_RTMO_RCFO}.
The unambiguous TMO and CFO data is randomly distributed within
$\left[ \frac{0\text{m}}{c},\frac{625\text{m}}{c} \right]$
and
$\left[ -\pi ,\pi  \right]$.
In particular, we set $c{{\tau }_{o,0}} = 205.6\text{m}$.

\subsection{Sensing Performance: Evaluation for Algorithm \ref{Algorithm1}}\label{V-B}

\begin{figure}[!t]
\centering
\subfigure[]
{\includegraphics[width=3.2in]{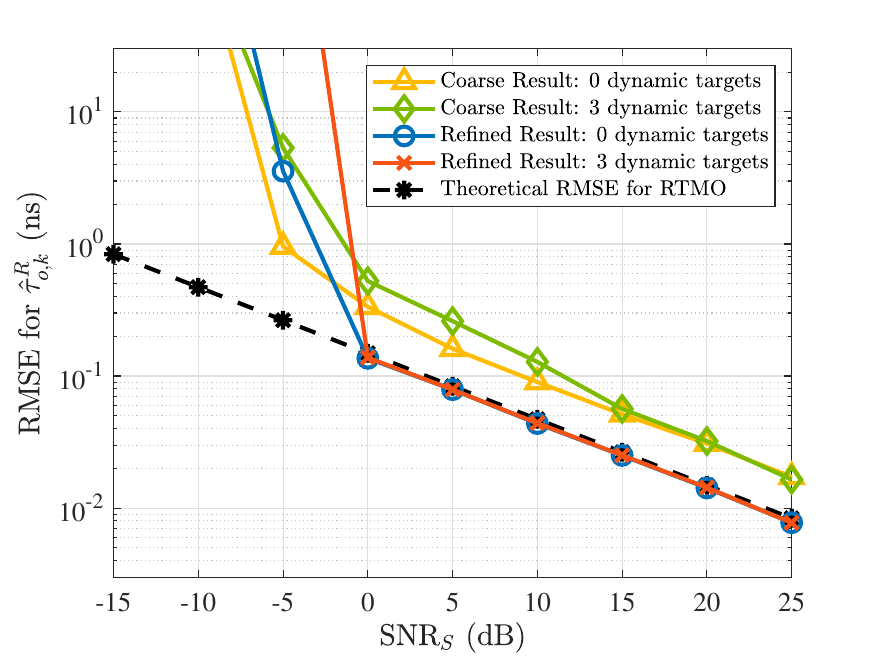}}
\subfigure[]
{\includegraphics[width=3.2in]{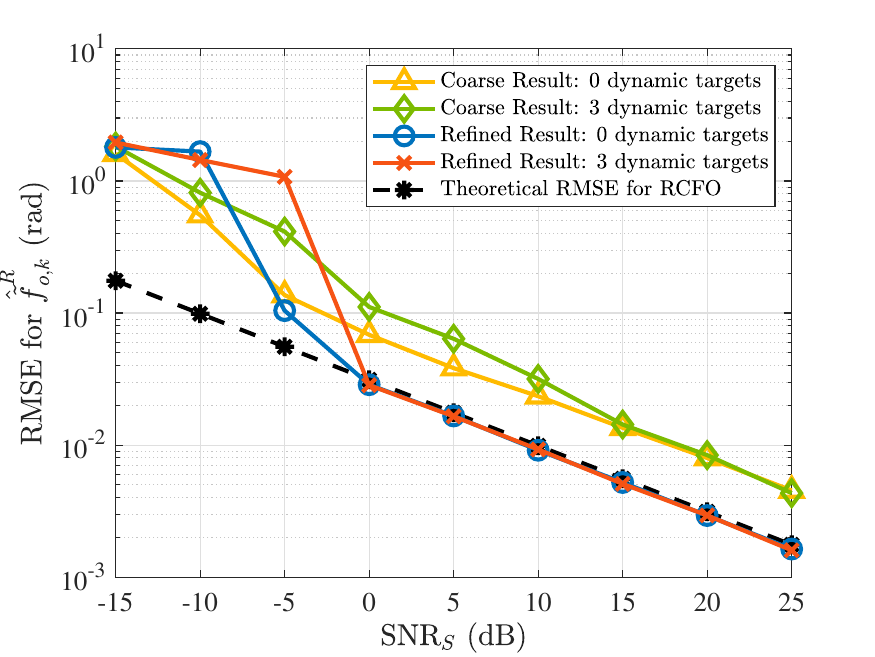}}
\caption{Results of Algorithm \ref{Algorithm1}: the RMSEs of (a) $\hat{\tau}^{R}_{o, k}$ and (b) $\hat{\tilde{f}}^{R}_{o,k}$ under different $\text{SNR}_S$ averaged over 1000 trials for each SNR$_S$ value.}
\label{Fig_simu1_2}
\end{figure}

First, we evaluate Algorithm \ref{Algorithm1} with $M = 4$.
Fig. \ref{Fig_simu1_2} illustrates the RMSE of $\hat{\tau}^{R}_{o, k}$ and $\hat{\tilde{f}}^{R}_{o,k}$ versus $\text{SNR}_S$.
In both figures, the refined estimation achieves a lower RMSE than the coarse estimation.
In coarse estimation, the RMSE of `3 dynamic targets' is higher than that of `0 dynamic targets' when SNR$_S < 15$dB.
In refined estimation, the RMSE is independent of the number of dynamic targets, which demonstrates the high efficacy of the refined estimator.
In addition, it is worth noting that the RMSE of the refined estimate declines faster than the coarse estimate with the increment of $\text{SNR}_S$.
This is because the noise in coarse estimation is no longer Gaussian distribution when the SNR is sufficiently high, referring to (\ref{4})-(\ref{MLE_fomular_RTMO_RCFO}).
Moreover, it can be observed from Fig. \ref{Fig_simu1_2} that the proposed Algorithm \ref{Algorithm1} works when $\text{SNR}_S\ge 0\text{dB}$.
This indicates that our algorithm is still effective when the signal power for static objects is only about 10dB higher than that for dynamic targets.
Finally, the RMSE of the refined estimate approaches to the theoretical RMSE provided in Proposition \ref{propo_var} for both
$\hat{\tau}^{R}_{o, k}$ and $\hat{\tilde{f}}^{R}_{o,k}$,
validating the theoretical results.

\begin{figure}[!t]
\centering
\subfigure[]
{\includegraphics[width=3.2in]{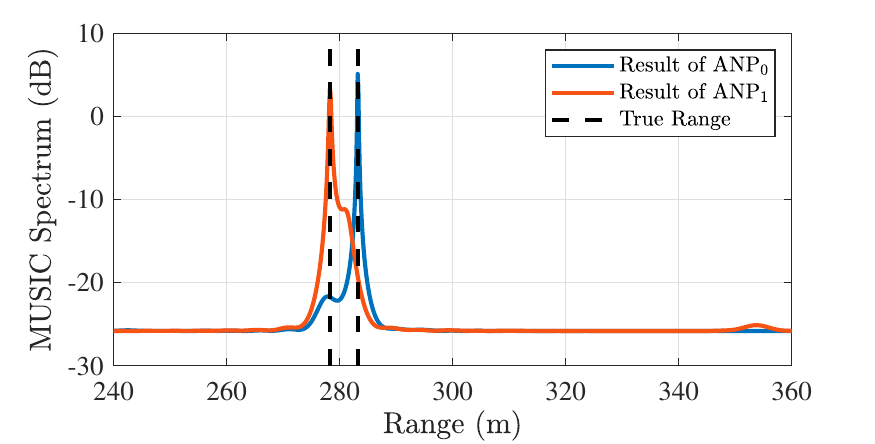}}
\subfigure[]
{\includegraphics[width=3.2in]{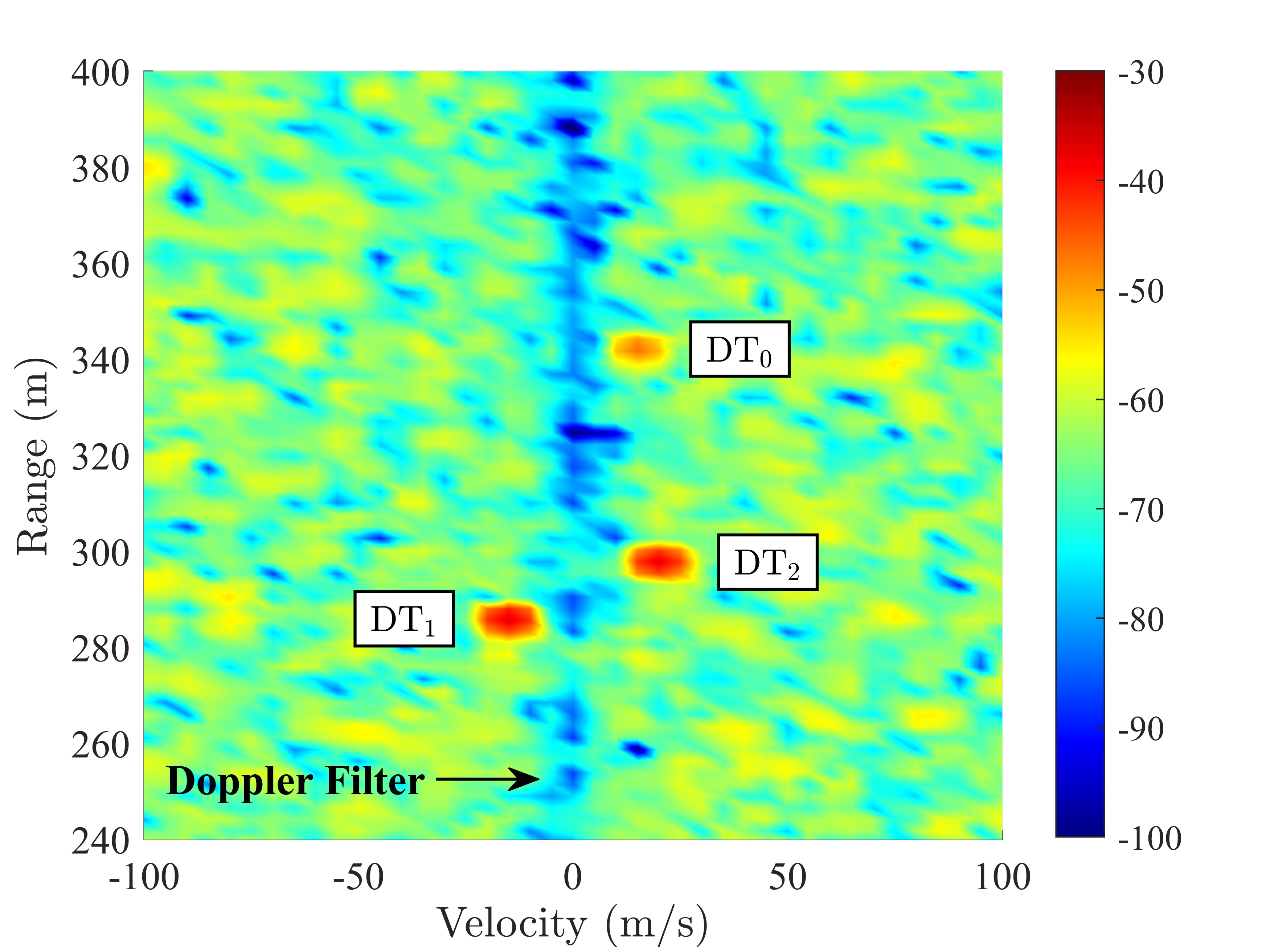}}
\caption{Results of dynamic target and anchor point estimation after RTMO and RCFO compensation, with
$c{{\tau }_{o,0}} = 205.6\text{m}$, exemplified by $\text{SN}{{\text{R}}_{S}}=4.13\text{dB}$.
(a) Result of identifying anchor points in range MUSIC spectrum; (b) Result after compensating and Doppler-domain filtering.}
\label{Fig_simu_2_RD_Spectrum}
\end{figure}

\subsection{Sensing Performance: Evaluation for Anchor Point Identification and Parameter Estimation}

We first present the results for anchor point identification.
The SNR of static objects is set to SNR$_S = 4.13\text{dB}$.
We obtain the $f_D = 0$ data from $\tilde{\tilde{y}}_{n,k}^{\left( m \right)}$ (\ref{second_compensate}) at each receiving antenna, and perform the spatial-smoothing-based MUSIC algorithm to calculate the noise subspace about the range and AOA.

Utilizing the known AOA of two anchor points, we obtain two range MUSIC spectrums with respect to the two AOAs,
shown in Fig. \ref{Fig_simu_2_RD_Spectrum}(a).
As can be seen, in addition to the anchor points, there are other peaks in the range MUSIC spectrum, which belong to the static objects.
According to A2, the power of the anchor points is higher than that of the static objects,
so the SNR of the anchor points' signal is higher.
When estimating the range using the MUSIC algorithm, a higher SNR will produce better orthogonality, leading to a higher peak in MUSIC spectrum.
Thus, we select the peaks with the maximum amplitude in Fig. \ref{Fig_simu_2_RD_Spectrum}(a) as the anchor points,
as shown in Fig. \ref{Fig_simu_2_RD_Spectrum}(a).
After identifying through MUSIC spectrum, we estimate the range and AOD of each anchor point, and the results are given in Table \ref{table1}.
The AOD estimates are unbiased, while the range estimates are offset by the reference TMO, i.e., $c{{\tau }_{o,0}} = 205.6\text{m}$.

Next, we remove static signals using a Doppler filter with transfer function $H\left( z \right)=\frac{1-0.9{{z}^{-1}}}{1-{{z}^{-1}}}$, as described in Section \ref{IV-A},
and the range-Doppler spectrum result is shown in Fig. \ref{Fig_simu_2_RD_Spectrum}(b).
It can be observed that the signals with
zero Doppler are entirely eliminated,
and the dynamic targets are all retained.
Then, we estimate the range and AOA of dynamic targets using the DFT-based algorithm, and the bias and RMSE results are given in Table \ref{table1}, averaging over 1000 trials.
Note that the dynamic targets and anchor points have the same biased values of range, i.e., $c{{\tau }_{o,0}} = 205.6\text{m}$.
This is used as the basis for removing $c{{\tau }_{o,0}}$ using Algorithm \ref{Algorithm2}.
In addition, the deviation of the estimated AOA and AOD is less than $10^{-4}$ degree, and the deviation of unbiased range estimate (without considering the reference TMO $c{{\tau }_{o,0}} = 205.6\text{m}$) is less than $1\text{mm}$.

Fig. \ref{Fig_simu_compare} compares the proposed algorithm and CACC
in terms of the estimation accuracy of
the velocity and relative range.
The term $\text{SNR}_D$ denotes the power ratio between dynamic path and noise.
The simulation setting is based on Section \ref{V-A}.
Since there is no LOS between the UE and BS, we cannot obtain the absolute range,
so we calculate the RMSE of the relative range for both algorithms.
The results show that the velocity and the relative range estimates of the proposed method are more accurate than that of CACC,
which proves the effectiveness of the proposed Algorithm \ref{Algorithm1}.

\begin{figure}[!t]
\centering
\includegraphics[width=3.2in]{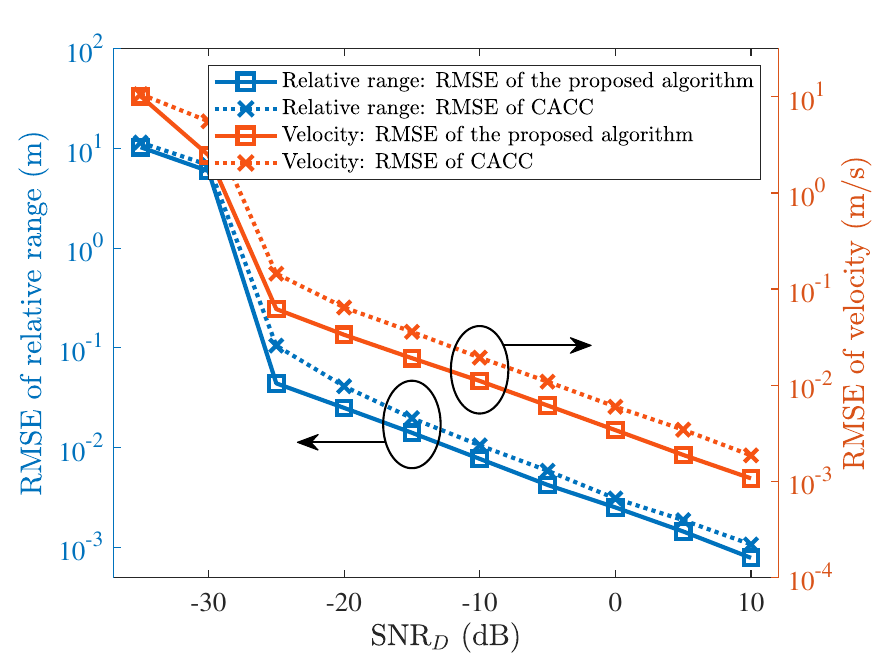}
\caption{Results of RMSE comparison between the proposed algorithm and CACC.}
\label{Fig_simu_compare}
\end{figure}

\subsection{Sensing Performance: Evaluation for Algorithm \ref{Algorithm2}}

Next, Algorithm \ref{Algorithm2} is evaluated with the assumption that all the anchor points are successfully identified in the range-AOA spectrum.
For the first evaluation, we use the output of Algorithm \ref{Algorithm1} (given in Table \ref{table1}) as the input of Algorithm \ref{Algorithm2} to locate UE and dynamic targets.
The results are plots in Fig. \ref{Fig_simu}, averaged over 1000 trials.
Clearly, all the estimate points fell precisely on the correspondence UE and dynamic targets in Fig. \ref{Fig_simu}.
These results clearly demonstrate that the proposed model in Section \ref{IV} is able to effectively overcome the impact of ${{\tau }_{o,0}}$ as well as the absence of LOS.

\begin{figure}[!t]
\centering
{\includegraphics[width=3.2in]{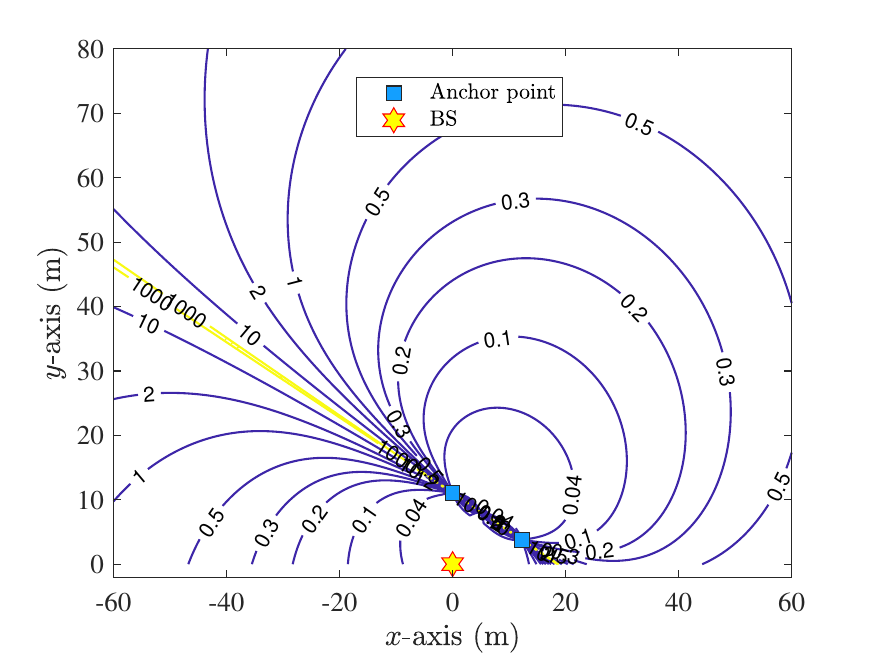}}
\caption{Localization accuracy of the UE at different locations using two anchor points.
Parameters are provided in Table \ref{table1} and the result is given in meters.}
\label{Fig_counter_map}
\end{figure}

To evaluate the impact of UE's location on the localization accuracy, we vary its locations and present the results in Fig. \ref{Fig_counter_map}, using two anchor points.
Note that the localization accuracy for the UE is defined as $\sqrt{\operatorname{var}\left\{ {{{\hat{p}}}_{x}} \right\}+\operatorname{var}\left\{ {{{\hat{p}}}_{y}} \right\}}$.
Clearly, the localization accuracy is infinite once the location of UE is collinear with two anchor points,
which validates Proposition \ref{propo_angle_and_something}.
Employing a third anchor point is able to solve this problem.
Moreover, the farther away from the anchor points, the lower the localization accuracy.
In practical applications, the accuracy of the area of interest can be improved by selecting the anchor points with optimal positions.

\begin{table}[]
\caption{Estimation Results of Dynamic Targets and Anchor Points}\label{table1}
\centering
\begin{tabular}{cccccc}
\toprule
\toprule
                           & ANP$_0$      &  ANP$_1$           & DT$_0$       & DT$_1$        & DT$_2$ \\ \midrule
\makecell[c]{Bistatic range (bias, mm)\\ (Add $c{{\tau }_{o,0}}$)}         & $<1$        &$<1$      & $<1$         & $<1$          & $<1$         \\
AOD (bias, degree)      & $<10^\text{-4}$ & $<10^\text{-4}$   & -----        & -----         & -----     \\
AOA (bias, degree)       & -----       & -----     & $<10^\text{-4}$  & $<10^\text{-4}$   &  $<10^\text{-4}$         \\
Bistatic range (RMSE, mm) & 5.26        & 5.58         & 27.99        & 9.96          & 12.03      \\
AOD (RMSE, degree)     & 0.06       & 0.06       & -----        & -----         & -----        \\
AOA (RMSE, degree)     & -----       & -----        & 0.19         & 0.29         &  0.09       \\
\bottomrule
\end{tabular}
\end{table}

\begin{figure}[!t]
\centering
\subfigure[]
{\includegraphics[width=3.2in]{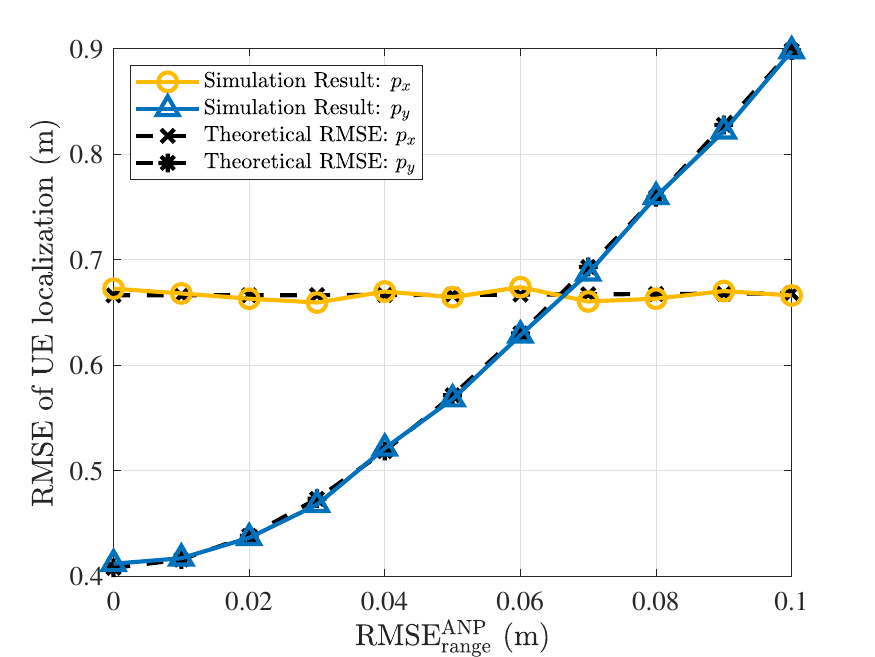}}
\subfigure[]
{\includegraphics[width=3.2in]{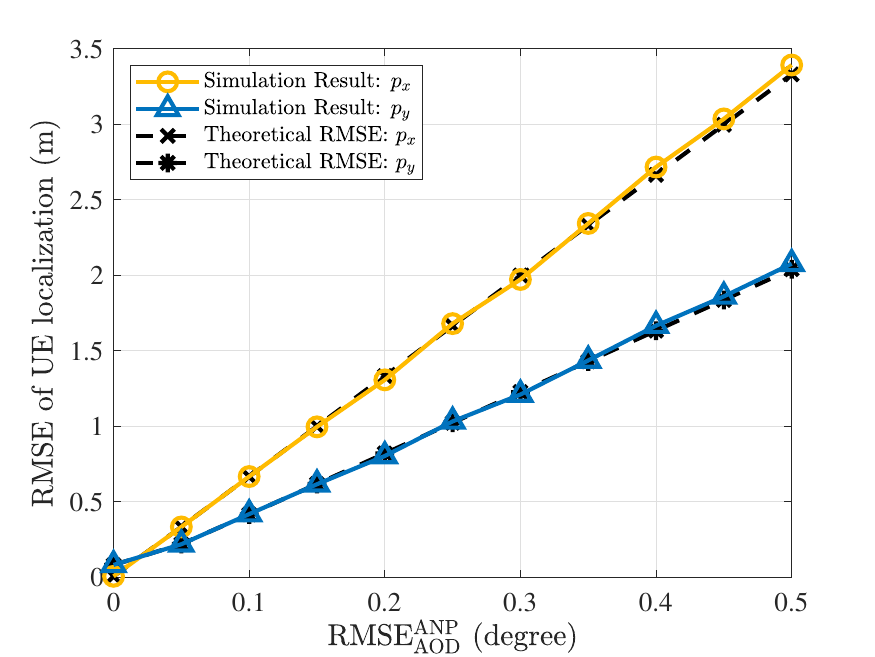}}
\caption{RMSE results of UE localization under different RMSEs of anchor point.
(a) RMSE results under different RMSE$_\text{range}^\text{ANP}$, with RMSE$_\text{AOD}^{\text{ANP}} = 0.1{}^\circ $;  (b) RMSE results under different RMSE$_\text{AOD}^\text{ANP}$, with RMSE$_\text{range}^{\text{ANP}} = 0.01\text{m}$.}
\label{Fig_simu_3_RMSE}
\end{figure}

Fig. \ref{Fig_simu_3_RMSE} shows the RMSE results of UE localization compared with theoretical RMSEs.
The theoretical RMSEs are given by (\ref{1st_in_vc}).
In particular, all the estimation results are unbiased.
Both figures show that the simulated results match the theoretical RMSEs well under different
RMSE$_\text{ range}^\text{ ANP}$ and RMSE$_\text{ AOD}^\text{ ANP}$.
Compared with Fig. \ref{Fig_simu_3_RMSE}(b), it can also be seen that
the localization RMSE is more sensitive to the RMSE$_\text{ AOD}^\text{ ANP}$ values of anchor points.
Interestingly, the UE localization RMSE is nearly zero when RMSE$_\text{ AOD}^\text{ ANP}$ equals zero,
which means a more accurate AOD estimation for anchor points leads to more precise UE localization.

\begin{figure}[!t]
\centering
{\includegraphics[width=3.2in]{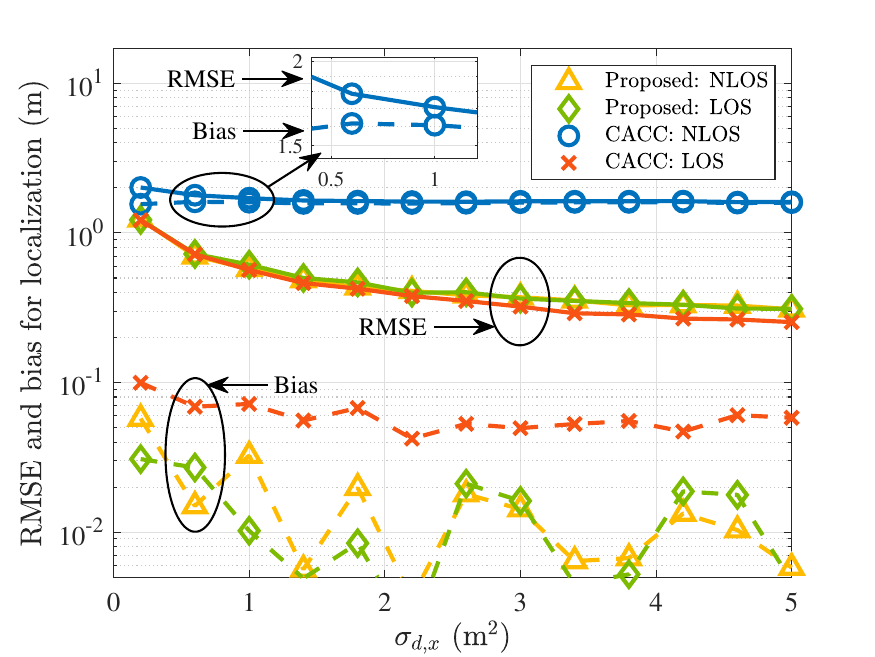}}
\caption{Localization performance of CACC and the proposed scheme.}
\label{Fig_simu_compare_all_algorithm}
\end{figure}

Fig. \ref{Fig_simu_compare_all_algorithm}
compares the localization performance between CACC and the proposed algorithms (cascading Algorithm 1 and Algorithm 2) for different
reflecting factors of DT1.
The simulation setting is described in Section \ref{V-A}.
For the LOS scenario, we added a LOS path signal into the echo and the power attenuation follows ${\frac{{{\lambda }^{2}}}{{{\left( 4\pi {{R}_{\text{LOS}}} \right)}^{2}}}}$, where ${{R}_{\text{LOS}}}$ denotes the LOS path range.
Fig. \ref{Fig_simu_compare_all_algorithm} shows that CACC and the proposed scheme achieve similar RMSEs and are both unbiased in the LOS case.
Thus, our method achieves nearly the same localization accuracy
without requiring to know the position of the UE.
In NLOS scenario, the proposed scheme still has the same accuracy as in LOS scenario and is also unbiased.
However, the CACC method fails with significant bias in the NLOS scenario since no LOS path serves as the reference path for localization.

\section{Conclusion}\label{VI_conclusion}
We have presented a novel anchor-point assisted uplink sensing scheme for ISAC systems with clock asynchronism between transmitter and receiver.
Our scheme can be applied to more general scenarios where the transmitter location is unknown to the receiver and there is no dominating LOS path between them.
The scheme consists of two main algorithms.
Algorithm \ref{Algorithm1} estimates RTMO and RCFO with respect to those at a reference snapshot.
Algorithm \ref{Algorithm2} pinpoints
the anchor points from the static objects, locates the UE and then locates dynamic targets.
We analytically show that
two anchor points are generally sufficient, and three non-collinear anchor points can guarantee to work.
The simulation results
demonstrate that our proposed scheme
can effectively locate the dynamic targets and UE, without requiring the LOS path or knowing the UE position,
and is very promising for practical applications.

Using anchor points,
the proposed scheme can evidently extend the sensing field-of-view of the BS,
especially in urban areas where line-of-sight paths are often blocked by buildings and plants.
For future work,
we can optimize their deployment/selection by taking
into consideration the environment layout, to achieve a high probability of the LOS path for the cascaded link
UE-Anchor-BS due to minimal constraints on the selection or installing of anchor points.

\begin{appendices}

\section{Proof of Proposition \ref{propo_matrix}}\label{Appendix_A}
Take the RCFO as an example.
The formula ${{\mathbf{g}}_{c}}=\frac{\partial }{\partial \mathbf{{c}}_{1}^{\left( m \right)}}\mathcal{L}\left( \mathbf{{c}}_{1}^{\left( m \right)},\bm{{\kappa} }_{1}^{\left( m \right)} \right)={{\mathbf{0}}_{\left( K-1 \right)\times 1}}$
holds while (\ref{iteration_algorithm}) converges.
For ease of derivation, we define
$\mathbf{w}_{c}^{\left( m \right)}=\exp \left( j\mathbf{{c}}_{1}^{\left( m \right)} \right) \in \mathbb{C}^{\left(K - 1\right) \times 1}$,
$\mathbf{w}_{\kappa }^{\left( m \right)}=\exp \left( j\bm{{\kappa} }_{1}^{\left( m \right)} \right) \in \mathbb{C}^{\left(K - 1\right) \times 1}$,
and
$ \mathbf{w}_{n}^{\left( m \right)}=\mathbf{w}_{c}^{\left( m \right)}\odot {{\left( \mathbf{w}_{\kappa }^{\left( m \right)} \right)}^{-n}}\in \mathbb{C}^{\left(K - 1\right) \times 1}$.
Based on A1,
each element in $\mathbf{\tilde{y}}_{s,n}^{\left( m \right)}$ is large.
The impact of $ \mathbf{{z}}_{n}^{\left( m \right)}$ on $ \mathbf{{w}}_{c}^{\left( m \right)}$
can be expressed by the first-order approximation of $ \mathbf{{w}}_{c}^{\left( m \right)}$ in the point $\mathbf{\tilde{y}}_{s,n}^{\prime\left( m \right)}$,
where
$\mathbf{{\tilde{y}}}_{s,n}^{\prime\left( m \right)}\in {{\mathbb{C}}^{\left( K-1 \right)\times 1}}$
is defined
from
$ \mathbf{\tilde{y}}_{s,n}^{\left( m \right)}={{\left[ {\tilde{y}}_{s,n,0}^{\left( m \right)},\left( \mathbf{{\tilde{y}}}_{s,n}^{\prime\left( m \right)} \right)^{T} \right]}^{T}}$:
\begin{align}\label{Appen_formu_2}
\nonumber \Delta \mathbf{{w}}_{c}^{\left( m \right)} & \approx \sum\limits_{n=0}^{N-1}{\frac{\partial \mathbf{{w}}_{c}^{ \left( m \right)}}{\partial {{\left( \mathbf{\tilde{y}}_{s,n}^{\prime\left( m \right)} \right)}^{T}}}\mathbf{{z}}_{n}^{\left( m \right)}}\\
&\approx -{{\left[ \frac{\partial {{\mathbf{g}}_{c}}}{\partial {{\left( \mathbf{{w}}_{c}^{\left( m \right)} \right)}^{T}}} \right]}^{-1}}\sum\limits_{n=0}^{N-1}{\frac{\partial {{\mathbf{g}}_{c}}}{\partial {{\left( \mathbf{\tilde{y}}_{s,n}^{\prime\left( m \right)} \right)}^{T}}}\mathbf{{z}}_{n}^{\left( m \right)}},
\end{align}
where the implicit differentiation is used.

The derivation of $\Delta \mathbf{{w}}_{\kappa }^{\left( m \right)}$ is identical to the process for deriving (\ref{Appen_formu_2}), and we can get the same result as (\ref{Appen_formu_2}).
With $\mathbf{\tilde{y}}_{s,n}^{\left( m \right)}={{\mathbf{1}}_{K\times 1}}h_{s,n}^{\left( m \right)}$,
we obtain four partial derivative results:
\begin{align}\label{Appen_formu_5}
\nonumber \frac{\partial {{\mathbf{g}}_{c}}}{\partial {{\left( \mathbf{\tilde{y}}_{s,n}^{\prime\left( m \right)} \right)}^{T}}} &= j{{\left( h_{s,n}^{\left( m \right)} \right)}^{*}}\left[ {{\mathbf{1}}_{\left( K-1 \right)\times 1}},{{\mathbf{1}}_{K-1}}-K{{\mathbf{I}}_{K-1}} \right] \\
 \nonumber \frac{\partial {{\mathbf{g}}_{\kappa }}}{\partial {{\left( \mathbf{\tilde{y}}_{s,n}^{\prime\left( m \right)} \right)}^{T}}} &=-jn{{\left( h_{s,n}^{\left( m \right)} \right)}^{*}}\left[ {{\mathbf{1}}_{\left( K-1 \right)\times 1}},{{\mathbf{1}}_{K-1}}-K{{\mathbf{I}}_{K-1}} \right]\\
\nonumber \frac{\partial {{\mathbf{g}}_{c}}}{\partial {{\left( \mathbf{{w}}_{c}^{ \left( m \right)} \right)}^{T}}} &=-j \sum\limits_{n=0}^{N-1}{{{\left| h_{s,n}^{\left( m \right)} \right|}^{2}}} \left({{\mathbf{1}}_{K-1}}-K{{\mathbf{I}}_{K-1}}\right)\\
\frac{\partial {{\mathbf{g}}_{\kappa }}}{\partial {{\left( \mathbf{{w}}_{\kappa }^{\left( m \right)} \right)}^{T}}}&=-j \sum\limits_{n=0}^{N-1}{{{n}^{2}}{{\left| h_{s,n}^{\left( m \right)} \right|}^{2}}} \left({{\mathbf{1}}_{K-1}}-K{{\mathbf{I}}_{K-1}}\right).
\end{align}
Substituting (\ref{Appen_formu_5}) into (\ref{Appen_formu_2}) and using ${\left({{\mathbf{1}}_{K-1}}-K{{\mathbf{I}}_{K-1}}\right)^{-1}}=-{{K}^{-1}}\left[ {{\mathbf{1}}_{K-1}}+{{\mathbf{I}}_{K-1}} \right]$ yields:
\begin{align}\label{Appen_formu_6}
\Delta \mathbf{{w}}_{c}^{\left( m \right)}\approx \frac{\sum\limits_{n=0}^{N-1}{\tilde{\bm{\mu }}_{n}^{\left( m \right)}}}{\sum\limits_{n=0}^{N-1}{{{\left| h_{s,n}^{\left( m \right)} \right|}^{2}}}}, \ \Delta \mathbf{{w}}_{\kappa }^{\left( m \right)}\approx \frac{-\sum\limits_{n=0}^{N-1}{n\tilde{\bm{\mu }}_{n}^{\left( m \right)}}}{\sum\limits_{n=0}^{N-1}{{{n}^{2}}{{\left| h_{s,n}^{\left( m \right)} \right|}^{2}}}},
\end{align}
where $\tilde{\bm{\mu }}_{n}^{\left( m \right)}={{\left( h_{s,n}^{\left( m \right)} \right)}^{*}}\left( \mathbf{{z}}_{n}^{\prime\left( m \right)}- \Delta \tilde{y}_{n,0}^{\left( m \right)}\mathbf{1}_{\left(K - 1\right) \times 1} \right) \in \mathbb{C}^{\left( K - 1 \right) \times 1}$.
Essentially, $z_{n,0}^{\left( m \right)}\mathbf{1}_{\left(K - 1\right) \times 1}$ can be regarded as ${{{c}}_{0}^{\left( m \right)}}$ and ${{\kappa}_{0}^{\left( m \right)}}$.
Based on Remark \ref{remark_coarse_RTMO_RCFO}, we only focus on the relative TMO and relative CFO, so that $z_{n,0}^{\left( m \right)}\mathbf{1}_{\left(K - 1\right) \times 1}$ can be omitted.
We redefine
$\bm{\mu }_{n}^{\left( m \right)} \triangleq {{\left( h_{s,n}^{\left( m \right)} \right)}^{*}} \mathbf{{z}}_{n}^{\prime\left( m \right)} \in \mathbb{C}^{\left( K - 1 \right) \times 1}$
and replace $\tilde{\bm{\mu }}_{n}^{\left( m \right)}$ with $\bm{\mu }_{n}^{\left( m \right)}$ in (\ref{Appen_formu_6}).

Next, we perform the first-order approximation with respect to phase for calculating $\Delta \mathbf{c}_{1}^{\left( m \right)}$ and $\Delta \bm{\kappa }_{1}^{\left( m \right)}$.
Note that the phase of $\mathbf{{w}}_{c}^{ \left( m \right)}$ and $\mathbf{{w}}_{\kappa }^{ \left( m \right)}$
is small after the coarse compensation.
We can obtain the following results:
\begin{align}\label{Appen_formu_7}
\nonumber \Delta \mathbf{c}_{1}^{\left( m \right)}&=\angle \left\{ \mathbf{{w}}_{c}^{\left( m \right)}+\Delta \mathbf{{w}}_{c}^{ \left( m \right)} \right\}-\angle \left\{ \mathbf{{w}}_{c}^{ \left( m \right)} \right\}  \approx \operatorname{Im}\left\{ \Delta \mathbf{{w}}_{c}^{ \left( m \right)} \right\} \\
 \Delta \bm{\kappa }_{1}^{\left( m \right)}&=\angle \left\{ \mathbf{{w}}_{\kappa }^{ \left( m \right)}+\Delta \mathbf{{w}}_{\kappa }^{ \left( m \right)} \right\}-\angle \left\{ \mathbf{{w}}_{\kappa }^{ \left( m \right)} \right\} \approx \operatorname{Im}\left\{ \Delta \mathbf{{w}}_{\kappa }^{ \left( m \right)} \right\}.
\end{align}
Substituting (\ref{Appen_formu_6}) into (\ref{Appen_formu_7}),
proposition \ref{propo_matrix} is thus proved.

\section{Proof of Proposition \ref{propo_var}}\label{Appendix_AABB}
Use $\operatorname{var}\left\{ \Delta \mathbf{c}_{1}^{\left( m \right)} \right\} $ as an example.
We define two real numbers first:  $\bar{h}_{s,n}^{\left( m \right)}=\operatorname{Re}\left\{ h_{s,n}^{\left( m \right)} \right\}$
and
$ \overset{\scriptscriptstyle\smile}{h}_{s,n}^{\left( m \right)}=\operatorname{Im}\left\{ h_{s,n}^{\left( m \right)} \right\}$.
Substituting $\bar{h}_{s,n}^{\left( m \right)}$
and
$ \overset{\scriptscriptstyle\smile}{h}_{s,n}^{\left( m \right)}$
into (\ref{propos_formula}) yields:
$ \operatorname{Im}\left\{ \bm{\mu }_{n}^{\left( m \right)} \right\} =\bar{h}_{s,n}^{\left( m \right)}\operatorname{Im}\left\{ \mathbf{z}_{n}^{\left( m \right)} \right\}-\overset{\scriptscriptstyle\smile}{h}_{s,n}^{\left( m \right)}\operatorname{Re}\left\{ \mathbf{z}_{n}^{\left( m \right)} \right\}$.
Then, $\operatorname{var}\left\{ \Delta \mathbf{c}_{1}^{\left( m \right)} \right\}$ is derived in (\ref{popopo_2_de}), given at the top of the next page.
\begin{figure*}[ht]
\begin{align}\label{popopo_2_de}
 \nonumber \operatorname{var}\left\{ \Delta \mathbf{c}_{1}^{\left( m \right)} \right\}&={{{\left( \sum\limits_{n=0}^{N-1}{{{\left| h_{s,n}^{\left( m \right)} \right|}^{2}}} \right)}^{-2}}}\sum\limits_{n=0}^{N-1}{{{\left( \bar{h}_{s,n}^{\left( m \right)} \right)}^{2}}\operatorname{var}\left\{ \operatorname{Im}\left\{ \mathbf{z}_{n}^{\prime\left( m \right)} \right\} \right\}+{{\left( \overset{\scriptscriptstyle\smile}{h}_{s,n}^{\left( m \right)} \right)}^{2}}\operatorname{var}\left\{ \operatorname{Re}\left\{ \mathbf{z}_{n}^{\prime\left( m \right)} \right\} \right\}} \\
 & =\frac{\sigma _{n}^{2}}{2}{{{\left( \sum\limits_{n=0}^{N-1}{{{\left| h_{s,n}^{\left( m \right)} \right|}^{2}}} \right)}^{-2}}}{{\mathbf{1}}_{\left( K-1 \right)\times 1}}\sum\limits_{n=0}^{N-1}{\left[ {{\left( \bar{h}_{s,n}^{\left( m \right)} \right)}^{2}}+{{\left( \overset{\scriptscriptstyle\smile}{h}_{s,n}^{\left( m \right)} \right)}^{2}} \right]}=\sigma _{n}^{2}{{{\left( \sum\limits_{n=0}^{N-1}{{{\left| h_{s,n}^{\left( m \right)} \right|}^{2}}} \right)}^{-1}}}{{\mathbf{1}}_{\left( K-1 \right)\times 1}}.
\end{align}
\normalsize
\hrulefill
\vspace*{4pt}
\end{figure*}
In the second step of (\ref{popopo_2_de}), the independence of the real and imaginary parts of the noise is exploited.
The term $\operatorname{var}\left\{ \Delta \bm{\kappa }_{1}^{\left( m \right)} \right\}$ can be similarly derived.
Proposition \ref{propo_var} is thus proved.

\section{Proof of Proposition \ref{propo_target_fft}}\label{Appendix_B}
In (\ref{Appen_formu_6}), ${{\mathcal{E}}_{\Delta }}$ can be expressed by $ \Delta w_{c,k}^{\left( m \right)}$ and $ \Delta w_{\kappa ,k}^{\left( m \right)} $,
where
$\Delta w_{c,k}^{\left( m \right)}={{\left[ \Delta \mathbf{w}_{c}^{\left( m \right)} \right]}_{k}}$
and
$\Delta w_{\kappa ,k}^{\left( m \right)}={{\left[ \Delta \mathbf{w}_{\kappa }^{\left( m \right)} \right]}_{k}}$:
\begin{align}\label{APPb_formu_one}
 \nonumber &{{\mathcal{E}}_{\Delta }} \approx \frac{1}{NK}{{\mathcal{F}}_{2}}\left\{ \left( 1+\Delta w_{c,k}^{\left( m \right)} \right){{\left( 1+\Delta w_{\kappa ,k}^{\left( m \right)} \right)}^{-n}}-1 \right\} \\
 \nonumber & \approx \frac{1}{NK}{{\mathcal{F}}_{2}}\left\{ \Delta w_{c,k}^{\left( m \right)}-n\Delta w_{\kappa ,k}^{\left( m \right)} \right\} \\
 \nonumber & =\frac{1}{NK}\left[ {{\mathcal{F}}_{\tau }}\left\{ 1 \right\}{{\mathcal{F}}_{{{f}_{D}}}}\left\{ \Delta w_{c,k}^{\left( m \right)} \right\}-{{\mathcal{F}}_{\tau }}\left\{ n \right\}{{\mathcal{F}}_{{{f}_{D}}}}\left\{ \Delta w_{\kappa ,k}^{\left( m \right)} \right\} \right] \\
 & \approx \frac{1}{K}\delta \left( \tau  \right)\left[ {{\mathcal{F}}_{{{f}_{D}}}}\left\{ \Delta w_{c,k}^{\left( m \right)} \right\}\!-\!\frac{N-1}{2}{{\mathcal{F}}_{{{f}_{D}}}}\left\{ \Delta w_{\kappa ,k}^{\left( m \right)} \right\} \right],
\end{align}
where $\delta \left( 0 \right)=1$ and $\delta \left( \tau \right)=0$, $\tau\ne 0$,
and $ {{\mathcal{F}}_{\tau }}\left\{ 1 \right\}=N\delta \left( \tau  \right)$
and
${{\mathcal{F}}_{\tau }}\left\{ n \right\}=\sum\nolimits_{n=0}^{N-1}{n{{e}^{j2\pi n\tau \Delta f}}}\approx \frac{\left( N-1 \right)N}{2}\delta \left( \tau  \right)$.
Substituting (\ref{Appen_formu_6}) into (\ref{APPb_formu_one}) yields:
\begin{align}\label{AppB_deri_two}
\nonumber & \operatorname{var}\left\{ {{\mathcal{E}}_{\Delta }} \right\} \\
\nonumber &=\frac{\delta \left( \tau  \right)}{K}\sigma _{n}^{2}\left[ \frac{1}{\sum\limits_{\tilde{n}=0}^{N-1}{{{\left| h_{s,n}^{\left( m \right)} \right|}^{2}}}}+{{\left( \frac{N-1}{2} \right)}^{2}}\frac{1}{\sum\limits_{\tilde{n}=0}^{N-1}{{{n}^{2}}{{\left| h_{s,n}^{\left( m \right)} \right|}^{2}}}} \right] \\
 & \approx \frac{\delta \left( \tau  \right)}{K}\frac{\sigma _{n}^{2}}{\sum\limits_{\tilde{n}=0}^{N-1}{{{\left| h_{s,n}^{\left( m \right)} \right|}^{2}}}}\left( 1+\frac{3}{4} \right)= \frac{1}{NK}\frac{7\sigma _{n}^{2}}{4{{\mathcal{P}}_{s}}}\delta \left( \tau \right),
\end{align}
where ${\mathcal{P}_{s}}=\frac{1}{N}\sum\nolimits_{n=0}^{N-1}{{{\left| h_{s,n}^{\left( m \right)} \right|}^{2}}}$ denotes the static object power,
and $\operatorname{var}\left\{ {{\mathcal{F}}_{{{f}_{D}}}}\left\{ z_{n,k}^{\left( m \right)} \right\} \right\}=K\sigma _{n}^{2}$.

Next, we substitute (\ref{AppB_deri_two}) into (\ref{formula_fft_influence}) and use the standard deviation to express ${{\mathcal{E}}_{\Delta }}$.
Then, ${{\mathcal{E}}_{\Delta }}\circledast {{\mathcal{F}}_{2}}\left\{ \tilde{y}_{s,n,k}^{\left( m \right)} \right\}$ is given by:
\begin{align}\label{the_influence_formula_standard_deviation}
\nonumber {{\mathcal{E}}_{\Delta }}\circledast {{\mathcal{F}}_{2}}\left\{ \tilde{y}_{s,n,k}^{\left( m \right)} \right\}&=\sqrt{\frac{1}{NK}\frac{7\sigma _{n}^{2}}{4{{\mathcal{P}}_{s}}}}\delta \left( \tau  \right)\circledast \left[ NK{{Y}_{s}}\left( \tau  \right)\delta \left( {{f}_{D}} \right) \right] \\
 &= \sqrt{\frac{1}{NK}\frac{7\sigma _{n}^{2}}{4{{\mathcal{P}}_{s}}}}NK{{Y}_{s}}\left( \tau  \right),
\end{align}
where $\delta \left( 0 \right)=1$ and $\delta \left( f_D \right)=0$, $f_D\ne 0$, and
${{\mathcal{F}}_{2}}\left\{ \tilde{y}_{s,n,k}^{\left( m \right)} \right\} = NK{{Y}_{s}}\left( \tau  \right)\delta \left( {{f}_{D}} \right)$.
Similar to (\ref{the_influence_formula_standard_deviation}), the noise term is given by ${{\mathcal{F}}_{noise}}={{\mathcal{F}}_{2}}\left\{ z_{n,k}^{\left( m \right)} \right\}=NK{{\sigma }_{n}}$.
Finally, we compute the ratio between ${{\mathcal{E}}_{\Delta }}\circledast {{\mathcal{F}}_{2}}\left\{ \tilde{y}_{s,n,k}^{\left( m \right)} \right\}$ and ${{\mathcal{F}}_{noise}}$ to prove Proposition \ref{propo_target_fft}:
\begin{align}\label{final_result_of_AppB}
\frac{{{\mathcal{E}}_{\Delta }}\circledast {{\mathcal{F}}_{2}}\left\{ \tilde{y}_{s,n,k}^{\left( m \right)} \right\}}{{{\mathcal{F}}_{noise}}}=\sqrt{\frac{7}{NK}}\frac{{{Y}_{s}}\left( \tau  \right)}{2\sqrt{{{\mathcal{P}}_{s}}}} \ll 1.
\end{align}
Note that ${{\mathcal{E}}_{\Delta }}\circledast {{\mathcal{F}}_{2}}\left\{ \tilde{y}_{s,n,k}^{\left( m \right)} \right\}$ is much lower than ${{\mathcal{F}}_{noise}}$.
Proposition \ref{propo_target_fft} is thus proved.

\end{appendices}

\bibliographystyle{IEEEtran}
\bibliography{Ancho}

\begin{thebibliography}{10}
\providecommand{\url}[1]{#1}
\csname url@samestyle\endcsname
\providecommand{\newblock}{\relax}
\providecommand{\bibinfo}[2]{#2}
\providecommand{\BIBentrySTDinterwordspacing}{\spaceskip=0pt\relax}
\providecommand{\BIBentryALTinterwordstretchfactor}{4}
\providecommand{\BIBentryALTinterwordspacing}{\spaceskip=\fontdimen2\font plus
\BIBentryALTinterwordstretchfactor\fontdimen3\font minus
  \fontdimen4\font\relax}
\providecommand{\BIBforeignlanguage}[2]{{%
\expandafter\ifx\csname l@#1\endcsname\relax
\typeout{** WARNING: IEEEtran.bst: No hyphenation pattern has been}%
\typeout{** loaded for the language `#1'. Using the pattern for}%
\typeout{** the default language instead.}%
\else
\language=\csname l@#1\endcsname
\fi
#2}}
\providecommand{\BIBdecl}{\relax}
\BIBdecl

\bibitem{9585321}
J.~A. Zhang, M.~L. Rahman, K.~Wu, X.~Huang, Y.~J. Guo, S.~Chen, and J.~Yuan,
  ``Enabling joint communication and radar sensing in mobile networks a
  survey,'' \emph{IEEE Commun. Surveys Tuts.}, vol.~24, no.~1, pp. 306--345,
  2022.

\bibitem{1011453411834}
\BIBentryALTinterwordspacing
S.~Li, Z.~Liu, Y.~Zhang, Q.~Lv, X.~Niu, L.~Wang, and D.~Zhang, ``Wiborder:
  Precise {Wi-Fi} based boundary sensing via through-wall discrimination,''
  \emph{Proc. ACM Interact. Mob. Wearable Ubiquitous Technol.}, vol.~4, no.~3,
  sep 2020. [Online]. Available: \url{https://doi.org/10.1145/3411834}
\BIBentrySTDinterwordspacing

\bibitem{9737357}
F.~Liu, Y.~Cui, C.~Masouros, J.~Xu, T.~X. Han, Y.~C. Eldar, and S.~Buzzi,
  ``Integrated sensing and communications: Toward dual-functional wireless
  networks for {6G} and beyond,'' \emph{IEEE J. Sel. Areas Commun.}, vol.~40,
  no.~6, pp. 1728--1767, 2022.

\bibitem{1011453351279}
\BIBentryALTinterwordspacing
Y.~Zeng, D.~Wu, J.~Xiong, E.~Yi, R.~Gao, and D.~Zhang, ``Farsense: Pushing the
  range limit of {WiFi}-based respiration sensing with {CSI} ratio of two
  antennas,'' \emph{Proc. ACM Interact. Mob. Wearable Ubiquitous Technol.},
  vol.~3, no.~3, sep 2019. [Online]. Available:
  \url{https://doi.org/10.1145/3351279}
\BIBentrySTDinterwordspacing

\bibitem{9606921}
C.~Li, N.~Raymondi, B.~Xia, and A.~Sabharwal, ``Outer bounds for a joint
  communicating radar (comm-radar): The uplink case,'' \emph{IEEE Trans.
  Commun.}, vol.~70, no.~2, pp. 1197--1213, 2022.

\bibitem{5393298}
C.~R. Berger, B.~Demissie, J.~Heckenbach, P.~Willett, and S.~Zhou, ``Signal
  processing for passive radar using {OFDM} waveforms,'' \emph{IEEE J. Sel.
  Topics Signal Process.}, vol.~4, no.~1, pp. 226--238, 2010.

\bibitem{8376990}
Z.~Abu-Shaban, X.~Zhou, T.~Abhayapala, G.~Seco-Granados, and H.~Wymeersch,
  ``Performance of location and orientation estimation in {5G} mmwave systems:
  Uplink vs downlink,'' in \emph{Proc. IEEE Wireless Commun. Netw. Conf.},
  2018, pp. 1--6.

\bibitem{9848428}
J.~A. Zhang, K.~Wu, X.~Huang, Y.~J. Guo, D.~Zhang, and R.~W. Heath,
  ``Integration of radar sensing into communications with asynchronous
  transceivers,'' \emph{IEEE Commun. Mag.}, pp. 1--7, 2022.

\bibitem{8443427}
W.~Yuan, N.~Wu, B.~Etzlinger, Y.~Li, C.~Yan, and L.~Hanzo,
  ``Expectation-maximization-based passive localization relying on asynchronous
  receivers: Centralized versus distributed implementations,'' \emph{IEEE
  Trans. Commun.}, vol.~67, no.~1, pp. 668--681, 2019.

\bibitem{9617146}
G.~Wen, H.~Song, F.~Gao, Y.~Liang, and L.~Zhu, ``Target localization in
  asynchronous distributed {MIMO} radar systems with a cooperative target,''
  \emph{IEEE Trans. Wireless Commun.}, vol.~21, no.~6, pp. 4098--4113, 2022.

\bibitem{2018Widar}
\BIBentryALTinterwordspacing
K.~Qian, C.~Wu, Y.~Zhang, G.~Zhang, and Y.~Liu, ``Widar2.0: Passive human
  tracking with a single {Wi-Fi} link,'' \emph{Proc. 16th Annu. Int. Conf.
  Mobile Syst., Appl. Serv.}, pp. 350--361, 2018. [Online]. Available:
  \url{https://doi.org/10.1145/3210240.3210314}
\BIBentrySTDinterwordspacing

\bibitem{9349171}
Z.~Ni, J.~A. Zhang, X.~Huang, K.~Yang, and J.~Yuan, ``Uplink sensing in
  perceptive mobile networks with asynchronous transceivers,'' \emph{IEEE
  Trans. on Signal Process.}, vol.~69, pp. 1287--1300, 2021.

\bibitem{Added_123777}
Y.~Zeng, D.~Wu, J.~Xiong, E.~Yi, R.~Gao, and D.~Zhang, ``Farsense: Pushing the
  range limit of {WiFi}-based respiration sensing with {CSI} ratio of two
  antennas,'' in \emph{Proc. ACM Interact. Mob. Wearable Ubiquitous Technol.},
  vol.~3, no.~3, 2019.

\bibitem{9645160}
D.~Wu, Y.~Zeng, R.~Gao, S.~Li, Y.~Li, R.~C. Shah, H.~Lu, and D.~Zhang,
  ``{WiTraj}: Robust indoor motion tracking with {WiFi} signals,'' \emph{IEEE
  Trans. Mobile Comput.}, vol.~22, no.~5, pp. 3062--3078, 2023.

\bibitem{10262009}
X.~Chen, Z.~Feng, J.~A. Zhang, X.~Yuan, and P.~Zhang, ``Kalman filter-based
  sensing in communication systems with clock asynchronism,'' \emph{IEEE Trans.
  Commun.}, vol.~72, no.~1, pp. 403--417, 2024.

\bibitem{8515231}
R.~Mendrzik, H.~Wymeersch, G.~Bauch, and Z.~Abu-Shaban, ``Harnessing {NLOS}
  components for position and orientation estimation in {5G} millimeter wave
  {MIMO},'' \emph{IEEE Trans. Wireless Commun.}, vol.~18, no.~1, pp. 93--107,
  2019.

\bibitem{8839839}
C.~E. O¡¯Lone, H.~S. Dhillon, and R.~M. Buehrer, ``Single-anchor localizability
  in {5G} millimeter wave networks,'' \emph{IEEE Wireless Commun. Lett.},
  vol.~9, no.~1, pp. 65--69, 2020.

\bibitem{1433260}
K.-F. Ssu, C.-H. Ou, and H.~Jiau, ``Localization with mobile anchor points in
  wireless sensor networks,'' \emph{IEEE Trans. Veh. Technol.}, vol.~54, no.~3,
  pp. 1187--1197, 2005.

\bibitem{9301354}
M.~Kolakowski, V.~Djaja-Josko, and J.~Kolakowski, ``Static {LiDAR} assisted
  {UWB} anchor nodes localization,'' \emph{IEEE Sensors J.}, vol.~22, no.~6,
  pp. 5344--5351, 2022.

\bibitem{5GNR211}
3GPP TS 38.211, Physical channels and modulation, V15.2.0, January 2023.

\bibitem{8928165}
X.~Lin, J.~Li, R.~Baldemair, J.-F.~T. Cheng, S.~Parkvall, D.~C. Larsson,
  H.~Koorapaty, M.~Frenne, S.~Falahati, A.~Grovlen, and K.~Werner, ``{5G New
  Radio}: Unveiling the essentials of the next generation wireless access
  technology,'' \emph{IEEE Commun. Stand. Mag.}, vol.~3, no.~3, pp. 30--37,
  2019.

\bibitem{9779639}
H.~Yin, N.~Li, J.~Guo, J.~Zhu, and X.~She, ``{NR} coverage enhancements for
  {PUSCH},'' \emph{IEEE Commun. Mag.}, vol.~60, no.~7, pp. 36--42, 2022.

\bibitem{1055282}
D.~Rife and R.~Boorstyn, ``Single tone parameter estimation from discrete-time
  observations,'' \emph{IEEE Trans. Inf. Theory}, vol.~20, no.~5, pp. 591--598,
  1974.

\bibitem{ArraySignalProcessing}
H.~L.~V. Trees, \emph{Optimum Array Processing. Detection, Estimation, and
  Modulation Theory, Part {IV}}, John Wiley \& Sons, Inc., New York, USA, 2002.

\bibitem{9393557}
H.~Sun, L.~G. Chia, and S.~G. Razul, ``Through-wall human sensing with {WiFi}
  passive radar,'' \emph{IEEE Trans. Aerosp. Electron. Syst.}, vol.~57, no.~4,
  pp. 2135--2148, 2021.

\bibitem{7833233}
L.~Zheng and X.~Wang, ``Super-resolution delay-{Doppler} estimation for {OFDM}
  passive radar,'' \emph{IEEE Trans. Signal Process.}, vol.~65, no.~9, pp.
  2197--2210, 2017.

\bibitem{RadarHandbook}
M.~I. Skolnik, \emph{Radar Handbook, Third Edition}.\hskip 1em plus 0.5em minus
  0.4em\relax New York, USA: The McGraw-Hill Companies, 2008.

\end{thebibliography}

\vfill

\end{document}